\documentclass[reqno]{amsart}

\usepackage{amssymb,latexsym,amsfonts,amsmath}
\usepackage{graphicx}
\usepackage{tikz}
\usetikzlibrary{automata}
\usepackage{dsfont}
\usepackage{diagrams}

\newtheorem{theorem}{Theorem}[section]
\newtheorem{lemma}[theorem]{Lemma}

\newtheorem{corollary}[theorem]{Corollary}

\newtheorem{definition}[theorem]{Definition}
\newtheorem{example}[theorem]{Example}

\newtheorem{remark}[theorem]{Remark}
\numberwithin{equation}{section}

\newcommand{\ve}{\varepsilon}

\newcommand{\q}{\mathsf q}

\newcommand{\boxspan}{\mathit{span}}
\newcommand{\params}{{\mathsf{q}}}

\newcommand{\R}{{\mathbb{R}}}

\newcommand{\Ze}{{\mathbb Z}}

\newcommand{\N}{{\mathbb{N}}}

\newcommand{\argmax}{\textrm{arg}\max}
\newcommand{\argmin}{\textrm{arg}\min}

\DeclareMathOperator{\diff}{d}
\newcommand{\ra}{\rightarrow}
\newcommand{\sigalg}{\mathcal{F}}
\newcommand{\filtration}{\mathds{F}}

\newcommand{\ul}{\underline}
\newcommand{\ol}{\overline}
\newcommand{\Let}{:=}
\newcommand{\EE}{\mathds{E}}
\newcommand{\PP}{\mathds{P}}
\newcommand{\traj}[3]{#1_{#2#3}}

\begin{document}

\begin{abstract}
Formal control synthesis approaches over stochastic systems have received significant attention in the past few years, in view of their ability to provide \emph{provably correct} controllers for complex \emph{logical} specifications in an \emph{automated} fashion. 
Examples of complex specifications of interest include properties expressed as formulae in linear temporal logic (LTL) or as automata on infinite strings. 
A general methodology to synthesize controllers for such properties resorts to \emph{symbolic abstractions} of the given stochastic systems. 
Symbolic models are discrete abstractions of the given concrete systems with the property that a controller designed on the abstraction can be refined (or implemented) into a controller on the original system. 
Although the recent development of techniques for the construction of symbolic models has been quite encouraging, 
the general goal of formal synthesis over stochastic control systems is by no means solved. 
A fundamental issue with the existing techniques is the known ``curse of dimensionality,'' 
which is due to the need to discretize state and input sets and that results in an exponential complexity over the number of state and input variables in the concrete system.  
In this work we propose a novel abstraction technique for incrementally stable stochastic control systems, 
which does not require state-space discretization but only input set discretization, 
and that can be potentially more efficient (and thus scalable) than existing approaches. 
We elucidate the effectiveness of the proposed approach by synthesizing a schedule for the coordination of two traffic lights under some safety and fairness requirements for a road traffic model. 
Further we argue that this 5-dimensional linear stochastic control system cannot be studied with existing approaches based on state-space discretization due to the very large number of generated discrete states.
\end{abstract}

\title[Towards Scalable Synthesis of Stochastic Control Systems]{Towards Scalable Synthesis of Stochastic Control Systems}

\author[M. Zamani]{Majid Zamani$^1$} 
\author[I. Tkachev]{Ilya Tkachev$^2$}
\author[A. Abate]{Alessandro Abate$^3$} 
\address{$^1$Department of Electrical and Computer Engineering, Technische Universit\"at M\"unchen, D-80290 Munich, Germany.}
\email{zamani@tum.de}
\urladdr{http://www.hcs.ei.tum.de}
\address{$^2$Delft Center for Systems and Control, Delft University of Technology, Mekelweg 2, 2628 CD, Delft, The Netherlands.}
\email{tkachev.ilya@gmail.com}
\urladdr{http://www.dcsc.tudelft.nl/$\sim$itkachev}
\address{$^3$Department of Computer Science, University of Oxford, Wolfson Building, Parks Road, Oxford, OX1 3QD, UK.}
\email{alessandro.abate@cs.ox.ac.uk}
\urladdr{https://www.cs.ox.ac.uk/people/alessandro.abate}

\maketitle

\section{Introduction}
\label{intro}
In the last decade many techniques have been developed providing controllers for control systems (both deterministic and, more recently, stochastic) in a \emph{formal} and \emph{automated} fashion against some complex \emph{logical} specifications. Examples of such specifications include properties expressed as formulae in linear temporal logic (LTL) or as automata on infinite strings \cite{katoen08}, and as such they are not tractable by classical techniques for control systems. A general scheme for providing such controllers is by leveraging \emph{symbolic models} of original concrete systems. Symbolic models are discrete abstractions of the original systems in which each symbol represents an aggregate of continuous variables. When such symbolic models exist for the concrete systems, one can leverage the algorithmic machinery for automated synthesis of discrete models \cite{InfGames,ComputeGames} to automatically synthesize discrete controllers which can be refined to hybrid controllers for the original systems.

The construction of symbolic models for \emph{continuous-time} non-probabilistic systems has been thoroughly investigated in the past few years. This includes results on the construction of approximately bisimilar symbolic
models for incrementally stable control systems \cite{majid4,pola}, switched systems \cite{girard2}, and control systems with disturbances \cite{pola1}, non-uniform abstractions of nonlinear systems over a finite-time horizon \cite{tazaki},
as well as the construction of sound abstractions based on the convexity of reachable sets \cite{gunther1}, feedback refinement relations \cite{ReissigWeberRungger15}, robustness margins \cite{liu1}, and for unstable control systems \cite{majid}. Recently, there have been some results on the construction of symbolic models for \emph{continuous-time} stochastic systems,
including the construction of finite Markov decision process approximations of linear stochastic control system, however without providing a quantitative relationship between
abstract and concrete model \cite{LAB09}, approximately bisimilar symbolic models for incrementally stable stochastic control systems \cite{majid8}, stochastic switched systems \cite{majid16}, and randomly switched stochastic systems \cite{majid11},
as well as sound abstractions for unstable stochastic control systems \cite{majid7}.  

Note that all the techniques provided in \cite{majid4,pola,girard2,pola1,tazaki,gunther1,ReissigWeberRungger15,liu1,majid,LAB09,majid8,majid11,majid7} are fundamentally based on the discretization of continuous states. Therefore, they suffer severely from the
curse of dimensionality due to gridding those sets, which is especially
irritating for models with high-dimensional state sets. In this work we propose a novel approach for the construction of approximately bisimilar symbolic models for incrementally stable stochastic control systems not requiring any state set discretization but only input set discretization. Therefore, it can be potentially more efficient than the proposed approaches in \cite{majid8} when dealing with higher dimensional stochastic control systems. We provide a theoretical comparison with the
approach in \cite{majid8} and a simple criterion that helps choosing the most suitable among two approaches (in terms of the sizes of the symbolic models) for a given stochastic control system. Another advantage of the technique proposed here is that it allows us to construct symbolic models with probabilistic output values, resulting in less conservative symbolic abstractions 
than those proposed in \cite{majid8,majid11,majid7} that allow for 
non-probabilistic output values exclusively.
We then explain how the proposed symbolic models with probabilistic output values can be used for synthesizing hybrid controllers enforcing logic specifications. The proposed approaches in \cite{majid16} also provide symbolic models with probabilistic output values and without any state set discretization. However, the results in \cite{majid16} are for stochastic switched systems rather than stochastic control systems as in this work and they do not provide any intuition behind the control synthesis over symbolic models with probabilistic output values. The effectiveness of the proposed results is illustrated by synthesizing a schedule for the coordination of two traffic lights
under some safety and fairness requirements for a model of road traffic which is a 5-dimensional linear stochastic control system. We also show that this example is not amenable to be dealt with the approaches proposed in \cite{majid8}. Although the main proposed results in this work are for incrementally stable stochastic control systems, the similar results
for incrementally stable non-probabilistic control systems can be recovered in the same
framework by simply setting the diffusion term to zero.

Alongside the relationship with and extension of \cite{majid16,majid8}, 
this paper provides a detailed and extended elaboration of the results first announced in \cite{majid10}, 
including the proofs of the main results, 
a detailed discussion on how to deal with probabilistic output values and a generalization of the corresponding result with no requirement on compactness, 
and finally discussing a new case study on road traffic control.

\section{Stochastic Control Systems}\label{sec2}
\subsection{Notation}\label{II.A}
The identity map on a set $A$ is denoted by $1_{A}$. The symbols $\N$, $\N_0$, $\Ze$, $\R$, $\R^+$, and $\R_0^+$ denote the set of natural, nonnegative integer, integer, real, positive, and nonnegative real numbers, respectively. The symbols $I_n$, $0_n$, and $0_{n\times{m}}$ denote the identity matrix, the zero vector, and the zero matrix in $\R^{n\times{n}}$, $\R^n$, and $\R^{n\times{m}}$, respectively. Given a vector \mbox{$x\in\mathbb{R}^{n}$}, we denote by $x_{i}$ the $i$--th element of $x$, and by $\Vert x\Vert$ the infinity norm of $x$, namely, \mbox{$\Vert x\Vert=\max\{|x_1|,|x_2|,...,|x_n|\}$}, where $|x_i|$ denotes the absolute value of $x_i$. Given a matrix $P=\{p_{ij}\}\in\R^{n\times{n}}$, we denote by $\text{Tr}({P})=\sum_{i=1}^np_{ii}$ the trace of $P$. We denote by $\lambda_{\min}(A)$ and $\lambda_{\max}(A)$ the minimum and maximum eigenvalues of a symmetric matrix $A$, respectively. The diagonal set $\Delta\subset\R^{n}\times\R^n$ is defined as: $\Delta=\left\{(x,x) \mid x\in \R^n\right\}$.

The closed ball centered at $x\in{\mathbb{R}}^{m}$ with radius $\lambda$ is defined by $\mathcal{B}_{\lambda}(x)=\{y\in{\mathbb{R}}^{m}\,|\,\Vert x-y\Vert\leq\lambda\}$. A set $B\subseteq \R^m$ is called a
{\em box} if $B = \prod_{i=1}^m [c_i, d_i]$, where $c_i,d_i\in \R$ with $c_i < d_i$ for each $i\in\{1,\ldots,m\}$.
The {\em span} of a box $B$ is defined as $\boxspan(B) = \min\{ | d_i - c_i| \mid i=1,\ldots,m\}$.
For a box $B\subseteq\R^m$ and $\mu \leq \boxspan(B)$,
define the $\mu$-approximation $[B]_\mu = [\R^m]_{\mu}\cap{B}$, where $[\R^m]_{\mu}=\{a\in \R^m\mid a_{i}=k_{i}\mu,k_{i}\in\mathbb{Z},i=1,\ldots,m\}$.
Note that $[B]_{\mu}\neq\varnothing$ for any $\mu\leq\boxspan(B)$.
Geometrically, for any $\mu\in{\mathbb{R}^+}$ with $\mu\leq\boxspan(B)$ and $\lambda\geq\mu$, the collection of sets
\mbox{$\{\mathcal{B}_{\lambda}(p)\}_{p\in [B]_{\mu}}$}
is a finite covering of $B$, i.e. \mbox{$B\subseteq\bigcup_{p\in[B]_{\mu}}\mathcal{B}_{\lambda}(p)$}.
We extend the notions of $\boxspan$ and {\em approximation} to finite unions of boxes as follows.
Let $A = \bigcup_{j=1}^M A_j$, where each $A_j$ is a box.
Define $\boxspan(A) = \min\{\boxspan(A_j)\mid j=1,\ldots,M\}$,
and for any $\mu \leq \boxspan(A)$, define $[A]_\mu = \bigcup_{j=1}^M [A_j]_\mu$.

Given a measurable function \mbox{$f:\mathbb{R}_{0}^{+}\rightarrow\mathbb{R}^n$}, the supremum of $f$ is denoted by
$\Vert f\Vert_{\infty}:=\text{(ess)sup}\{\Vert f(t)\Vert, t\geq0\}$. A continuous function \mbox{$\gamma:\mathbb{R}_{0}^{+}\rightarrow\mathbb{R}_{0}^{+}$} is said to belong to class $\mathcal{K}$ if it is strictly increasing and \mbox{$\gamma(0)=0$}; $\gamma$ is said to belong to class $\mathcal{K}_{\infty}$ if \mbox{$\gamma\in\mathcal{K}$} and $\gamma(r)\rightarrow\infty$ as $r\rightarrow\infty$. A continuous function \mbox{$\beta:\mathbb{R}_{0}^{+}\times\mathbb{R}_{0}^{+}\rightarrow\mathbb{R}_{0}^{+}$} is said to belong to class $\mathcal{KL}$ if, for each fixed $s$, the map $\beta(r,s)$ belongs to class $\mathcal{K}$ with respect to $r$ and, for each fixed nonzero $r$, the map $\beta(r,s)$ is decreasing with respect to $s$ and $\beta(r,s)\rightarrow 0$ as \mbox{$s\rightarrow\infty$}. We identify a relation \mbox{$R\subseteq A\times B$} with the map \mbox{$R:A \rightarrow 2^{B}$} defined by $b\in R(a)$ iff \mbox{$(a,b)\in R$}. Given a relation \mbox{$R\subseteq A\times B$}, $R^{-1}$ denotes the inverse relation defined by \mbox{$R^{-1}=\{(b,a)\in B\times A:(a,b)\in R\}$}. Given a finite sequence $S$, we denote by $\sigma:=(S)^\omega$ the infinite sequence generated by repeating $S$ infinitely, i.e. $\sigma:=SSSSS\ldots$.

\subsection{Stochastic control systems\label{II.B}}

Let $(\Omega, \sigalg, \PP)$ be a probability space endowed with a filtration $\filtration = (\sigalg_s)_{s\geq 0}$ satisfying the usual conditions of completeness and right continuity \cite[p.\ 48]{ref:KarShr-91}. Let $(W_s)_{s \ge 0}$ be a $p$-dimensional $\filtration$-adapted Brownian motion.

\begin{definition}
\label{Def_control_sys}A stochastic control system $\Sigma$ is a tuple $\Sigma=(\mathbb{R}^{n},\mathsf{U},\mathcal{U},f,\sigma)$, where
\begin{itemize}
\item $\mathbb{R}^{n}$ is the state space;
\item $\mathsf{U}\subseteq\R^m$ is a bounded input set;
\item $\mathcal{U}$ is a subset of the set of all measurable functions of time from $\R_0^+$ to $\mathsf{U}$;
\item $f:\R^n\times\mathsf{U}\rightarrow\R^n$ satisfies the following Lipschitz assumption: there exist constants $L_x,L_u\in\R^+$ such that: $\Vert f(x,u)-f(x',u')\Vert\leq L_x\Vert x-x'\Vert +  L_u\Vert u-u'\Vert$ for all $x,x'\in\R^n$ and all $u,u'\in\mathsf{U}$;
\item $\sigma:\R^n\rightarrow\R^{n\times{p}}$ satisfies the following Lipschitz assumption: there exists a constant $Z\in\R^+$ such that: $\Vert\sigma(x)-\sigma(x')\Vert\leq Z\Vert{x}-x'\Vert$ for all $x,x'\in\R^n$.\qed
\end{itemize}
\end{definition}


A continuous-time stochastic process \mbox{$\xi:\Omega \times\R_0^+\rightarrow \mathbb{R}^{n}$} is said to be a \textit{solution process} of $\Sigma$ if there exists $\upsilon\in\mathcal{U}$ satisfying the following stochastic differential equation (SDE) $\PP$-almost surely ($\PP$-a.s.)
\begin{equation}
\label{eq0}
	\diff \xi= f(\xi,\upsilon)\diff t+\sigma(\xi)\diff W_t,
\end{equation}
where $f$ is known as the \emph{drift} and $\sigma$ as the \emph{diffusion}. We also write $\xi_{a \upsilon}(t)$ to denote the value of the solution process at time $t\in\R_0^+$ under the input curve $\upsilon$ from initial condition $\xi_{a \upsilon}(0) = a$ $\PP$-a.s., in which $a$ is a random variable that is measurable in $\sigalg_0$. Let us emphasize that the solution process is unambiguously determined,
since the assumptions on $f$ and $\sigma$ ensure its existence and uniqueness \cite[Theorem 5.2.1, p.\ 68]{oksendal}.

\section{Incremental Stability}\label{sec3}
We recall a stability notion for stochastic control systems, introduced in \cite{majid8}, on which the main results presented in this work rely.

\begin{definition}
\label{dISS}
A stochastic control system $\Sigma$ is incrementally input-to-state stable in the $q\textsf{th}$ moment ($\delta$-ISS-M$_q$), where $q\geq1$, if there exist a $\mathcal{KL}$ function $\beta$ and a $\mathcal{K}_{\infty}$ function $\gamma$ such that for any $t\in{\mathbb{R}_0^+}$, any $\R^n$-valued random variables $a$ and $a'$ that are measurable in $\sigalg_0$, and any $\upsilon$, ${\upsilon}'\in\mathcal{U}$, the following condition is satisfied:
\begin{equation}
\EE \left[\left\Vert \xi_{a\upsilon}(t)-\xi_{a'{\upsilon}'}(t)\right\Vert^q\right] \leq\beta\left( \EE\left[ \left\Vert a-a' \right\Vert^q\right], t \right) + \gamma \left(\left\Vert{\upsilon} - {\upsilon}'\right\Vert_{\infty}\right). \label{delta_PISS}
\end{equation}
\end{definition}

It can be easily verified that a $\delta$-ISS-M$_q$ stochastic control system $\Sigma$ is $\delta$-ISS \cite{angeli} in the absence of any noise as in the following: 
\begin{equation}
\left\Vert \xi_{a\upsilon}(t)-\xi_{a'{\upsilon}'}(t)\right\Vert\leq\beta\left(\left\Vert a-a' \right\Vert, t \right) + \gamma \left(\left\Vert{\upsilon} - {\upsilon}'\right\Vert_{\infty}\right), \label{delta_ISS}
\end{equation}
for $a,a'\in\R^n$, some $\beta\in\mathcal{KL}$, and some $\gamma\in\mathcal{K}_\infty$.

Similar to the characterization of $\delta$-ISS in terms of the existence of so-called $\delta$-ISS Lyapunov functions in \cite{angeli}, one can describe $\delta$-ISS-M$_q$ in terms of the existence of so-called $\delta$-ISS-M$_q$ Lyapunov functions, as shown in \cite{majid8} and defined next.

\begin{definition}
\label{delta_PISS_Lya}
Consider a stochastic control system $\Sigma$
and a continuous function $V:\mathbb{R}^n\times\mathbb{R}^n\rightarrow\mathbb{R}_0^+$ that is twice continuously differentiable on
$\{\R^n\times\R^n\}\backslash\Delta$.
The function $V$ is called a $\delta$-ISS-M$_q$ Lyapunov function for $\Sigma$,
where $q\geq1$, if there exist $\mathcal{K}_{\infty}$ functions
$\underline{\alpha}$, $\overline{\alpha}$, $\rho$, and a constant $\kappa\in\mathbb{R}^+$, such that
\begin{itemize}
\item[(i)] $\ul{\alpha}$ (resp. $\ol \alpha$) is a convex (resp. concave) function;
\item[(ii)] for any $x,x'\in\mathbb{R}^n$,
$\underline{\alpha}\left(\Vert x-x'\Vert^q\right)\leq{V}(x,x')\leq\overline{\alpha}\left(\Vert x-x'\Vert^q\right)$;
\item[(iii)] for any $x,x'\in\mathbb{R}^n$, $x\neq x'$, and for any $u,u'\in\mathsf{U}$,
\begin{align*}
	\mathcal{L}^{u,u'} V(x, x')\Let& \left[\partial_xV~~\partial_{x'}V\right] \begin{bmatrix} f(x,u)\\f(x',u')\end{bmatrix}+\frac{1}{2} \text{Tr} \left(\begin{bmatrix} \sigma(x) \\ \sigma(x') \end{bmatrix}\left[\sigma^T(x)~~\sigma^T(x')\right] \begin{bmatrix}
\partial_{x,x} V & \partial_{x,x'} V \\ \partial_{x',x} V & \partial_{x',x'} V
\end{bmatrix}	\right)\\\notag \leq&-\kappa V(x,x')+\rho(\| u-u'\|),
\end{align*}
\end{itemize}
where $\mathcal{L}^{u,u'}$ is the infinitesimal generator associated to the process $V(\xi,\xi')$, 
and where $\xi$ and $\xi'$ are solution processes of the SDE \eqref{eq0} \cite[Section 7.3]{oksendal}. 
The symbols $\partial_x$ and $\partial_{x,x'}$ denote first- and second-order partial derivatives with respect to $x$ and $(x,x')$, respectively.
\qed
\end{definition}

Although condition $(ii)$ in the above definition implies that the growth rate of functions $\ol{\alpha}$ and $\ul{\alpha}$ is linear,
this requirement does not restrict the behavior of $\ol\alpha$ and $\ul\alpha$ to be linear on a compact subset of $\R^n$. 
Note that condition $(i)$ is not required in the context of non-probabilistic control systems for the corresponding $\delta$-ISS Lyapunov functions \cite{angeli}.  
The following theorem, borrowed from \cite{majid8}, describes $\delta$-ISS-M$_q$ in terms of the existence of $\delta$-ISS-M$_q$ Lyapunov functions.

\begin{theorem}
\label{the_Lya}
A stochastic control system $\Sigma$ is $\delta$-ISS-M$_q$ if it admits a $\delta$-ISS-M$_q$ Lyapunov function.
\qed
\end{theorem}

One can resort to available software tools, such as \textsf{SOSTOOLS} \cite{antonis},
to search for appropriate $\delta$-ISS-M$_q$ Lyapunov functions for systems $\Sigma$ of polynomial type. 
We refer the interested readers to the results in \cite{majid8} for the discussion of special instances where these functions can be easily computed, 
and limit ourselves to mention that, as an example, 
for linear stochastic control systems $\Sigma$ (with linear drift and diffusion terms), 
one can search for appropriate $\delta$-ISS-M$_q$ Lyapunov functions by solving a linear matrix inequality (LMI). 

\subsection{Noisy and noise-free trajectories}
In order to introduce the symbolic models in Subsection \ref{ssec:main.res} (Theorems \ref{main_theorem} and \ref{main_theorem2}) for a stochastic control system,
we need the following technical result, borrowed from \cite{majid8},
which provides an upper bound on the distance (in the $q\textsf{th}$ moment) between the solution process of $\Sigma$ and the solution of a derived non-probabilistic control system $\ol\Sigma$ obtained by disregarding the diffusion term $\sigma$. From now on, we use the notation $\ol\xi_{x\upsilon}$ to denote the solution of
$\ol\Sigma=(\mathbb{R}^{n},\mathsf{U},\mathcal{U},f,0_{n\times p})$\footnote{Here, we have abused notation by identifying $0_{n\times p}$ with the map $\sigma:x\rightarrow 0_p$ $\forall x\in\R^n$.}, starting from the non-probabilistic initial condition $x$ and under the input curve $\upsilon$, 
which satisfies the ordinary differential equation (ODE) $\dot{\ol\xi}_{x\upsilon}=f(\ol\xi_{x\upsilon},\upsilon)$. 
\begin{lemma}\label{lemma3}
Consider a stochastic control system $\Sigma$ such that $f(0_n,0_m) =0_n$ and $\sigma(0_n) = 0_{n\times{p}}$.
	Suppose that $q\geq2$ and that there exists a $\delta$-ISS-M$_q$ Lyapunov function $V$ for $\Sigma$ such that its Hessian is a positive semidefinite matrix in $\R^{2n\times2n}$ and $\partial_{x,x}{V}(x,x')\leq P$, for any $x,x'\in\R^n$, and some positive semidefinite matrix $P\in\R^{n\times n}$. Then for any $x\in\R^n$ and any $\upsilon\in\mathcal{U}$, we have
	\begin{align}\label{mismatch1}
	 	\EE \left[\left\Vert\traj{\xi}{x}{\upsilon}(t)-\traj{\ol \xi}{x}{\upsilon}(t)\right\Vert^q\right] \le h_x(t),	
	\end{align}
	where
	\begin{align}\notag
	h_x(t)&=\ul\alpha^{-1}\Bigg(\frac{1}{2}\left\Vert{\sqrt{P}}\right\Vert^2n\min\{n,p\}Z^2\mathsf{e}^{-\kappa t}\int_0^t\left(\beta\left(\left\Vert{x}\right\Vert^q,s\right)+\gamma\left(\sup_{u\in{\mathsf{U}}}\left\{\Vert{u}\Vert\right\}\right)\right)^{\frac{2}{q}}\mathsf{d}s\Bigg),
	\end{align}
	and where $Z$ is the Lipschitz constant, introduced in Definition \ref{Def_control_sys}, and $\beta$ is the $\mathcal{KL}$ function appearing in \eqref{delta_PISS}.\qed
\end{lemma}

It can be readily seen that the nonnegative-valued function $h_x$ tends to zero as $t \ra 0$, $t\ra+\infty$, or as $Z\ra0$, 
and is identically zero if the diffusion term is identically zero (i.e. $Z=0$, which is the case for $\ol\Sigma$). 
The interested readers are referred to \cite{majid8}, which provides results in line with that of Lemma \ref{lemma3} for (linear) stochastic control systems $\Sigma$ admitting a specific type of $\delta$-ISS-M$_q$ Lyapunov functions.

\section{Systems and Approximate Equivalence Relations}\label{symbolic}

\subsection{Systems}
We employ the abstract and general notion of ``system,'' as introduced in \cite{paulo}, to describe both stochastic control systems and their symbolic models. 
\begin{definition}
\label{system}
A system $S$ is a tuple
$S=(X,X_0,U,\rTo,Y,H),$
where
$X$ is a set of states (possibly infinite),
$X_0\subseteq X$ is a set of initial states (possibly infinite),
$U$ is a set of inputs (possibly infinite),
$\rTo\subseteq X\times U\times X$ is a transition relation,
$Y$ is a set of outputs, and
$H:X\rightarrow Y$ is an output map.\qed
\end{definition}

A transition \mbox{$(x,u,x')\in\rTo$} is also denoted by $x\rTo^ux'$. For a transition $x\rTo^ux'$, state $x'$ is called a \mbox{$u$-successor}, or simply a successor, of state $x$.
We denote by $\mathbf{Post}_{u}(x)$ the set of all \mbox{$u$-successors} of a state $x$.
For technical reasons, we assume that for any $x\in X$, there exists some $u$-successor of $x$, for some $u\in U$ ---
let us remark that this is always the case for the considered systems later in this paper.

A system $S$ is said to be
\begin{itemize}
\item \textit{metric}, if the output set $Y$ is equipped with a metric
$\mathbf{d}:Y\times Y\rightarrow\mathbb{R}_{0}^{+}$;
\item \textit{finite} (or \textit{symbolic}), if $X$ and $U$ are finite sets;
\item \textit{deterministic}, if for any state $x\in{X}$ and any input $u\in{U}$, $\left\vert\mathbf{Post}_{u}(x)\right\vert\leq1$.
\end{itemize}

For a system $S=(X,X_0,U,\rTo,Y,H)$ and given any initial state $x_0\in X_0$, a finite state run generated from $x_0$ is a finite sequence of transitions:
\begin{align}\label{run}
x_0\rTo^{u_0}x_1\rTo^{u_1}\cdots\rTo^{u_{n-2}}x_{n-1}\rTo^{u_{n-1}}x_n,
\end{align}
such that $x_i\rTo^{u_i}x_{i+1}$ for all $0\leq i<n$. A finite state run can be directly extended to an infinite state run as well. A finite output run is a sequence $\left\{y_0,y_1,\ldots,y_n\right\}$ such that there exists a finite state run of the form \eqref{run} with $y_i=H(x_i)$, for $i=0,\ldots,n$. A finite output run can also be directly extended to an infinite output run as well.

\subsection{Relations among systems}
\label{ssec:sys.rel}
We recall the notion of approximate (bi)simulation relation, introduced in \cite{girard},
which is cruicial when analyzing or synthesizing controllers for deterministic systems.

\begin{definition}\label{APSR}
Let \mbox{$S_{a}=(X_{a},X_{a0},U_{a},\rTo_{a},Y_a,H_{a})$} and
\mbox{$S_{b}=(X_{b},X_{b0},U_{b},\rTo_{b},Y_b,H_{b})$} be metric systems with the
same output sets $Y_a=Y_b$ and metric $\mathbf{d}$.
For $\varepsilon\in\mathbb{R}_0^{+}$,
a relation
\mbox{$R\subseteq X_{a}\times X_{b}$} is said to be an $\varepsilon$-approximate simulation relation from $S_{a}$ to $S_{b}$
if, for all $(x_a,x_b)\in R$, the following two conditions are satisfied:
\begin{itemize}
\item[(i)] \mbox{$\mathbf{d}(H_{a}(x_{a}),H_{b}(x_{b}))\leq\varepsilon$};
\item[(ii)] \mbox{$x_{a}\rTo_{a}^{u_a}x'_{a}$ in $S_a$} implies the existence of \mbox{$x_{b}\rTo_{b}^{u_b}x'_{b}$} in $S_b$ satisfying $(x'_{a},x'_{b})\in R$.
\end{itemize}
A relation $R\subseteq X_a\times X_b$ is said to be an $\varepsilon$-approximate bisimulation relation between $S_a$ and $S_b$
if $R$ is an $\varepsilon$-approximate simulation relation from $S_a$ to $S_b$ and
$R^{-1}$ is an $\varepsilon$-approximate simulation relation from $S_b$ to $S_a$.

System $S_{a}$ is $\varepsilon$-approximately simulated by $S_{b}$, or $S_b$ $\varepsilon$-approximately simulates $S_a$,
denoted by \mbox{$S_{a}\preceq_{\mathcal{S}}^{\varepsilon}S_{b}$}, if there exists an $\varepsilon$-approximate simulation relation $R$ from $S_a$ to $S_b$ such that:
\begin{itemize}
\item for every $x_{a0}\in{X_{a0}}$, there exists $x_{b0}\in{X_{b0}}$ with $(x_{a0},x_{b0})\in{R}$.
\end{itemize}
System $S_{a}$ is $\varepsilon$-approximately bisimilar to $S_{b}$, denoted by \mbox{$S_{a}\cong_{\mathcal{S}}^{\varepsilon}S_{b}$}, if there exists an $\varepsilon$-approximate bisimulation relation $R$ between $S_a$ and $S_b$ such that:
\begin{itemize}
\item for every $x_{a0}\in{X_{a0}}$, there exists $x_{b0}\in{X_{b0}}$ with $(x_{a0},x_{b0})\in{R}$;
\item for every $x_{b0}\in{X_{b0}}$, there exists $x_{a0}\in{X_{a0}}$ with $(x_{a0},x_{b0})\in{R}$.\qed
\end{itemize}
\end{definition}

\section{Symbolic Models for Stochastic Control Systems}\label{existence}

\subsection{Describing stochastic control systems as metric systems}
\label{ssec:fin.abstr}
In order to show the main results of the paper, we use the notion of system introduced above to abstractly represent a stochastic control system.
More precisely, given a stochastic control system $\Sigma$, 
we define an associated metric system $S(\Sigma)=(X,X_0,U,\rTo,Y,H),$
where:
\begin{itemize}
\item $X$ is the set of all $\R^n$-valued random variables defined on the probability space
$(\Omega,\sigalg,\PP)$;
\item $X_{0}$ is a subset of the set of $\R^n$-valued random variables that are measurable over $\sigalg_0$;
\item $U=\mathcal{U}$;
\item $x\rTo^{\upsilon} x'$ if $x$ and $x'$ are measurable in $\sigalg_{t}$ and $\sigalg_{t+\tau}$, respectively, for some $t \in \R^+_0$ and $\tau\in\R^+$, and
there exists a solution process $\xi:\Omega\times\R_0^+\rightarrow\R^n$ of $\Sigma$ satisfying $\xi(t) = x$ and $\xi_{x \upsilon}(\tau) = x'$ $\PP$-a.s.;
\item $Y=X$;
\item $H=1_{X}$.
\end{itemize}
We assume that the output set $Y$ is equipped with the metric $\mathbf{d}(y,y')=\left(\EE\left[\left\Vert y-y'\right\Vert^q\right]\right)^{\frac{1}{q}}$, for any $y,y'\in{Y}$ and some $q\geq1$. Let us remark that the set of states and inputs of $S(\Sigma)$ are uncountable and that $S(\Sigma)$ is a deterministic system in the sense of Definition \ref{system}, since (cf. Subsection \ref{II.B})
the solution process of $\Sigma$
is uniquely determined. Note that for the case of non-probabilistic control system $\ol\Sigma$, one obtains $S(\ol\Sigma)=(X,X_0,U,\rTo,Y,H)$, where $X=\R^n$, $X_0$ is a subset of $\R^n$, $U=\mathcal{U}$, $x\rTo^{\upsilon} x'$ iff $x'=\ol\xi_{x\upsilon}(\tau)$ for some $\tau\in\R^+$, $Y=X$, $H=1_{X}$, and the metric on the output set reduces to the natural Euclidean one: $\mathbf{d}(y,y')=\left\Vert y-y'\right\Vert$, for any $y,y'\in{Y}$.

Notice that, since the concrete system $S(\Sigma)$ is uncountably infinite, 
it does not allow for a straightforward discrete controller synthesis with the techniques in the literature \cite{InfGames,ComputeGames}.  
We are thus interested in finding a finite abstract system that is (bi)similar to the concrete system $S(\Sigma)$. 
In order to discuss approximate (bi)simulation relations between two metric systems,
they have to share the output space (cf. Definition \ref{APSR}).
System $S(\Sigma)$ inherits a classical trace-based semantics (cf. definition of output run after \eqref{run}) \cite{katoen08},
however 
the outputs of $S(\Sigma)$ (and necessarily those of any approximately (bi)similar one) are random variables. 
This fact is especially important due to the metric $\mathbf{d}$ that the output set is endowed with:
for any non-probabilistic point one can always find a non-degenerate random variable that is as close as desired to the original point in the metric $\mathbf{d}$. 

To further elaborate the discussion in the previous paragraph,
let us consider the following example.
Let $A\subset \R^n$ be a set (of non-probabilistic points). 
Consider a safety problem, 
formulated as the satisfaction of the LTL formula\footnote{We refer the interested readers to \cite{katoen08} for the formal semantic of the temporal formula $\square \varphi_A$ expressing the safety property over set $A$.} $\square \varphi_A$, where $\varphi_A$ is a label (or proposition) characterising the set $A$. 
Suppose that over the abstract system we are able to synthesize a control strategy that makes an output run of the abstraction satisfy $\square \varphi_A$.
Although the run would in general be consisting of random variables $y$,
the fact that $y\in A$ means that $y$ has a Dirac probability distribution centered at $y$,
that is $y\in Y$ is a degenerate random variable that can be identified with a point in $A\subset\R^n\subset Y$.
Note that since any non-probabilistic point can be regarded as a random variable with a Dirac probability distribution centered at that point,
$\R^n$ can be embedded in $Y$, which we denote as $\R^n\subset Y$ with a slight abuse of notation.
As a result, satisfying $\square \varphi_A$ precisely means that the output run of the abstraction indeed stays in the set $A\subset \R^n$ forever.
On the other hand,
suppose that the original system is $\ve$-approximate bisimilar to the abstraction.
If we want to interpret the result $\square \varphi_A$ obtained over the abstraction,
we can guarantee that the corresponding output run of the original system satisfies $\square \varphi_{A_\ve}$,
that is any output $y$ of the run of the original system is within $\ve$ $\mathbf{d}$-distance from the set $A$: $\mathbf{d}(y,A) = \inf_{a\in A}\mathbf{d}(y,a) \leq \ve$.
Note that although the original set $A\subset Y$ is a subset of $\R^n\subset Y$,
its $\ve$-inflation $A_\ve = \{y\in Y:\mathbf{d}(y,A)\leq \ve\}$ is not a subset of $\R^n$ anymore and hence contains non-degenerate random variables.
In particular, $A_\ve\neq \{y\in \R^n:\inf_{a\in A}\|y - a\|\leq \ve\}$ and is in fact bigger than the latter set of non-probabilistic points.
As a result, although satisfying $\square \varphi_{A_\ve}$ does not necessarily mean that a trajectory of $\Sigma$ always stays within some non-probabilistic set,
it means that the associated random variables always belong to $A_\ve$ and, hence, are close to the non-probabilistic set $A$ with respect to the $q$\textsf{th} moment metric.

We are now able to provide two versions of finite abstractions:
one whose outputs are always non-probabilistic points --
that is degenerate random variables,
elements of $\R^n\subset Y$,
and one whose outputs can be non-degenerate random variables.
Recall, however,
that in both cases the output set is still the whole $Y$ and the semantics is the same as for the original system $S(\Sigma)$.

\subsection{Main results}
\label{ssec:main.res}
This subsection contains the main contributions of the paper.
We show that for any $\delta$-ISS-M$_q$ (resp. $\delta$-ISS) stochastic control system $\Sigma$ (resp. non-probabilistic control system $\ol\Sigma$),
and for any precision level $\varepsilon\in\R^+$, we can construct a finite system that is $\varepsilon$-approximate bisimilar to $\Sigma$ (resp. $\ol\Sigma$) without any state set discretization.
The results in this subsection rely on additional assumptions on the model $\Sigma$ that are described next.
We restrict our attention to stochastic control systems $\Sigma$ with input sets $\mathsf{U}$ that are assumed to be finite unions of boxes (cf. Subsection \ref{II.A}).
We further restrict our attention to sampled-data stochastic control systems, where input curves belong to set $\mathcal{U}_\tau$, 
which contains exclusively curves that are constant over intervals of length $\tau\in\R^+$, i.e.
$$ \mathcal{U}_\tau=\Big\{\upsilon\in\mathcal{U}\,\,\vert\,\,\upsilon(t)=\upsilon((k-1)\tau), t\in[(k-1)\tau,k\tau[, k\in\N\Big\}.$$
Let us denote by $S_\tau(\Sigma)$ a sub-system of $S(\Sigma)$ obtained by selecting those transitions of $S(\Sigma)$ corresponding to solution processes of duration $\tau$ and to control inputs in $\mathcal{U}_\tau$.
This can be seen as the time discretization of $\Sigma$.
More precisely,
given a stochastic control system $\Sigma$ and the corresponding metric system $S(\Sigma)$,
we define a new associated metric system $$S_\tau(\Sigma)=\left(X_\tau,X_{\tau0},U_\tau,\rTo_\tau,Y_\tau,H_\tau\right),$$
where $X_\tau=X$, $X_{\tau0}=X_0$, $U_\tau=\mathcal{U}_\tau$, $Y_\tau=Y$, $H_\tau=H$, and
\begin{itemize}
\item $x_\tau\rTo^{\upsilon_\tau}_\tau{x'_\tau}$ if $x_\tau$ and $x'_\tau$ are measurable,
respectively, in $\sigalg_{k\tau}$ and $\sigalg_{(k+1)\tau}$ for some $k \in \N_0$, and
there exists a solution process $\xi:\Omega\times\R_0^+\rightarrow\R^n$ of $\Sigma$ satisfying $\xi(k\tau) = x_\tau$ and $\xi_{x_\tau\upsilon_\tau}(\tau) = x'_\tau$ $\PP$-a.s..
\end{itemize}

Similarly, one can define $S_\tau(\ol\Sigma)$ as the time discretization of $\ol\Sigma$. Notice that a finite state run
$$x_{0}\rTo^{\upsilon_{0}}_{\tau}x_{1}\rTo^{\upsilon_1}_{\tau}\cdots\rTo^{\upsilon_{N-1}}_{\tau} x_{N}$$ of $S_{\tau}(\Sigma)$, where $\upsilon_{i-1}\in\mathcal{U}_\tau$ and $x_i=\xi_{x_{i-1}\upsilon_{i-1}}(\tau)$ $\PP$-a.s. for $i=1,\ldots,N$, captures the solution process of $\Sigma$ at times $t=0,\tau,\ldots,N\tau$,
started from the initial condition $x_{0}$ and resulting from a control input $\upsilon$ obtained by the concatenation of the input curves
$\upsilon_{i-1}$ \big(i.e. $\upsilon(t)=\upsilon_{i-1}(t)$ for any $t\in [(i-1)\tau,i\,\tau[$\big), for $i=1,\ldots,N$.

\medskip

Let us proceed introducing two fully symbolic systems for the concrete model $\Sigma$.
Consider a stochastic control system $\Sigma$ and a tuple $\mathsf{q}=\left(\tau,\mu,N,x_s\right)$ of parameters, where $\tau$ is the sampling time, $\mu$ is the input set quantization, $N\in\N$ is a \emph{temporal horizon}, and $x_s\in\R^n$ is a \emph{source state}.
Given $\Sigma$ and $\mathsf{q}$, let us introduce the following two symbolic systems:
\begin{align}\notag
S_{\mathsf{q}}(\Sigma)&=(X_{\mathsf{q}},X_{\mathsf{q}0},U_{\mathsf{q}},\rTo_{\mathsf{q}},Y_{\mathsf{q}},H_{\mathsf{q}}),\\\notag
\ol{S}_{\mathsf{q}}(\Sigma)&=(X_{\mathsf{q}},X_{\mathsf{q}0},U_{\mathsf{q}},\rTo_{\mathsf{q}},Y_{\mathsf{q}},\ol{H}_{\mathsf{q}}),
\end{align}
consisting of:
\begin{itemize}
\item $X_{\mathsf{q}}=\big\{\left(u_1,\ldots,u_N\right)\in\overbrace{[\mathsf{U}]_\mu\times\cdots\times[\mathsf{U}]_\mu}^{N~\text{times}}\big\}$;
\item $X_{\params0}=X_{\mathsf{q}}$;
\item $U_{\mathsf{q}}=[\mathsf{U}]_\mu$;
\item $x_\params\rTo_{\mathsf{q}}^{u_{\mathsf{q}}}x'_{\mathsf{q}}$, where $x_\params=(u_1,u_2,\ldots,u_N)$, if and only if $x'_\params=(u_2,\ldots,u_N,u_\params)$;
\item $Y_{\mathsf{q}}$ is the set of all $\R^n$-valued random variables defined on the probability space $(\Omega,\sigalg,\PP)$;
\item $H_{\mathsf{q}}(x_\params)=\xi_{x_sx_\params}(N\tau)$ $\left(\ol{H}_{\mathsf{q}}(x_\params)=\ol\xi_{x_sx_\params}(N\tau)\right)$.
\end{itemize}

Note that the transition relation in $S_\params(\Sigma)$ admits a compact representation in the form of a \emph{shift operator}. 
We have abused notation by identifying $u_{\mathsf{q}}\in[\mathsf{U}]_\mu$ with the constant input curve with domain $[0,\tau[$ and value $u_{\mathsf{q}}$, 
and by identifying $x_\params\in [\mathsf{U}]_\mu^N$ with the concatenation of $N$ control inputs $u_i\in [\mathsf{U}]_\mu$ \big(i.e. $x_\params(t)=u_i$ for any $t\in[(i-1)\tau,i\tau[$\big) for $i=1,\ldots,N$. Notice that the proposed abstraction $S_\params(\Sigma)$ $\left(\text{resp.}~\ol{S}_{\params}(\Sigma)\right)$ is a deterministic system in the sense of Definition \ref{system}. Note that $H_\params$ and $\ol{H}_\params$ are mappings from a non-probabilistic point $x_\params$ to the random variable $\xi_{x_sx_\params}(N\tau)$ and to the one with a Dirac probability distribution centered at $\ol\xi_{x_sx_\params}(N\tau)$, respectively. 
Finally, note that in the case of a non-probabilistic control system $\ol\Sigma$, one obtains the symbolic system $\ol S_\params(\ol\Sigma)=(X_{\mathsf{q}},X_{\mathsf{q}0},U_{\mathsf{q}},\rTo_{\mathsf{q}},Y_{\mathsf{q}},\ol{H}_{\mathsf{q}})$, where $X_{\mathsf{q}}$, $X_{\params0}$, $U_{\mathsf{q}}$, $\rTo_{\mathsf{q}}$, and $\ol{H}_{\mathsf{q}}$ are the same as before, but where the output set reduces to $Y_\params=\R^n$.  

Note that the idea behind the definitions of symbolic models $S_\params(\Sigma)$ and $\ol{S}_\params(\Sigma)$ hinges on the $\delta$-ISS-M$_q$ property. Given an input $\upsilon\in\mathcal{U}$, all solution processes of $\Sigma$ under the input $\upsilon$ forget the mismatch between their initial conditions and converge to each other with respect to the $q\textsf{th}$ moment metric. Therefore, the longer the applied inputs are, the less relevant is the mismatch between initial conditions. Then, the fundamental idea of the introduced abstractions consists in taking the $N$ applied inputs as the state of the symbolic model. 

The control synthesis over $\ol{S}_\params(\Sigma)$ (resp. $\ol{S}_\params(\ol\Sigma)$) is simple as the outputs are non-probabilistic points, 
whereas for $S_\params(\Sigma)$ it is perhaps less intuitive. Hence, we discuss it in more details later in Subsection \ref{ssec:ltl.control}. 


\begin{example}
An example of an abstraction $S_\params(\Sigma)$ with $N=3$ and $U_{\mathsf{q}}=\left\{0,1\right\}$ is depicted in Figure \ref{fig3}, where the initial states are shown as targets of sourceless arrows. 
Note that, regardless of the size of the state set and of its dimension, 
$S_\params(\Sigma)$ only has eight possible states, namely: 
$$X_\params=\{(0,0,0),(0,0,1),(0,1,0),(0,1,1),(1,0,0),(1,0,1),(1,1,0),(1,1,1)\}.$$ 

\end{example}
\begin{figure}
\begin{center}
\scalebox{1}{
\begin{tikzpicture}[shorten >=1pt,node distance=3.50cm,auto]
\tikzstyle{state}=[state with output]
\tikzstyle{every state}=[draw=brown!100,very thick,fill=brown!50,minimum size=0.5cm]
\node[state,initial,initial text=,initial where=left] (x1) {$(0,0,0)$ \nodepart{lower} {$\xi_{x_s(0,0,0)}(3\tau)$}};
\node[state,initial,initial text=,initial where=above] (x2) [right of=x1] {$(0,0,1)$ \nodepart{lower} {$\xi_{x_s(0,0,1)}(3\tau)$}};
\node[state,initial,initial text=,initial where=above] (x3) [right of=x2] {$(0,1,1)$ \nodepart{lower} {$\xi_{x_s(0,1,1)}(3\tau)$}};
\node[state,initial,initial text=,initial where=right] (x4) [right of=x3] {$(1,1,1)$ \nodepart{lower} {$\xi_{x_s(1,1,1)}(3\tau)$}};
\node[state,initial,initial text=,initial where=below] (x5) [below of=x4] {$(1,1,0)$ \nodepart{lower} {$\xi_{x_s(1,1,0)}(3\tau)$}};
\node[state,initial,initial text=,initial where=below] (x6) [below of=x3] {$(1,0,1)$ \nodepart{lower} {$\xi_{x_s(1,0,1)}(3\tau)$}};
\node[state,initial,initial text=,initial where=below] (x7) [below of=x2] {$(0,1,0)$ \nodepart{lower} {$\xi_{x_s(0,1,0)}(3\tau)$}};
\node[state,initial,initial text=,initial where=below] (x8) [below of=x1] {$(1,0,0)$ \nodepart{lower} {$\xi_{x_s(1,0,0)}(3\tau)$}};
\path[->] (x1) edge node {1} (x2)
(x1) edge [loop above] node {0} (x1)
(x2) edge node {1} (x3)
(x2) edge node {0} (x7)
(x3) edge node {1} (x4)
(x3) edge node {0} (x5)
(x4) edge node {0} (x5)
(x4) edge [loop above] node {1} (x4)
(x5) edge node {1} (x6)
(x5) edge [bend right=-40] node {0} (x8)
(x6) edge node {0} (x7)
(x6) edge node {1} (x3)
(x7) edge node {0} (x8)
(x7) edge [bend right=-35] node {1} (x6)
(x8) edge node {0} (x1)
(x8) edge node {1} (x2);
\end{tikzpicture}
}
\end{center}
\caption {Example of abstraction $S_\params(\Sigma)$ with $N=3$ and $U_{\mathsf{q}}=\left\{0,1\right\}$. 
The lower part of the states are intended as labels, corresponding to their output values. 
Initial states are targets of sourceless arrows. }\label{fig3} 
\end{figure}
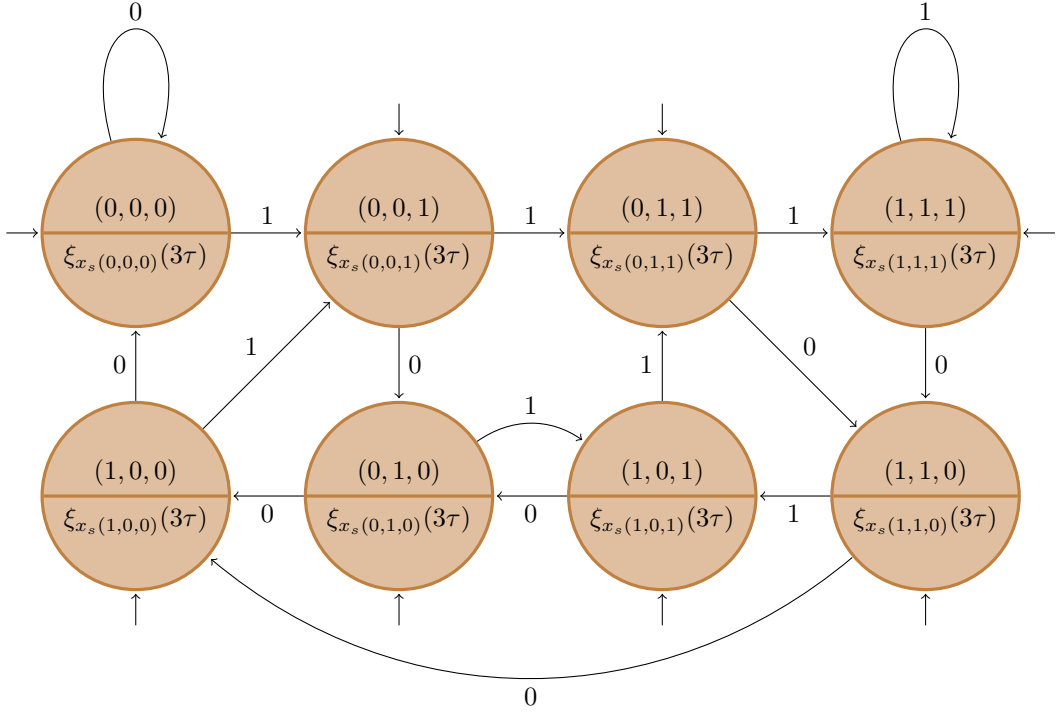

In order to obtain some of the main results of this work,
we raise an assumption on the $\delta$-ISS-M$_q$ Lyapunov function $V$ we will work with, as follows:
\begin{equation}\label{supplement}
\vert V(x,y)-V(x,z)\vert\leq\widehat\gamma(\Vert y-z\Vert),
\end{equation}
for any $x,y,z\in\R^n$, and for some $\mathcal{K}_\infty$ and concave function $\widehat\gamma$. As long as one is interested to work in a compact subset of $\R^n$, the function $\widehat\gamma$ in \eqref{supplement} can be readily computed. Indeed, for all $x,y,z\in \mathsf{D}$, where $\mathsf{D}\subset\R^n$ is compact, one can readily apply the mean value theorem to the function $y\rightarrow V(x,y)$ to get $$\left\vert V(x,y)-V(x,z)\right\vert\leq\widehat\gamma\left(\Vert y-z\Vert\right),~~\text{where}~~\widehat\gamma({r})=\left(\max_{x,y\in \mathsf{D}\backslash\Delta}\left\Vert\frac{\partial{V}(x,y)}{\partial{y}}\right\Vert\right)r.$$
In particular, for the $\delta$-ISS-M$_1$ Lyapunov function $V(x,x')\Let\sqrt{\left(x-x'\right)^TP\left(x-x'\right)}$, for some positive definite matrix $P\in\R^{n\times{n}}$ and for all $x,x'\in\R^n$,
one obtains $\widehat\gamma({r})= \frac{\lambda_{\max}\left(P\right)}{\sqrt{\lambda_{\min}\left(P\right)}} r$ \cite[Proposition 10.5]{paulo},
which satisfies \eqref{supplement} globally on $\R^n$. Note that for non-probabilistic control systems, the concavity assumption of $\widehat\gamma$ is not required. 

Before providing the main results of the paper, we need the following technical lemmas.

\begin{lemma}\label{lemma1}
Consider a stochastic control system $\Sigma$, admitting a $\delta$-ISS-M$_q$ Lyapunov function $V$, and consider its corresponding symbolic model $\ol{S}_\params(\Sigma)$. We have that 
\begin{align}\label{upper_bound}
\eta\leq&\left(\ul\alpha^{-1}\left(\mathsf{e}^{-\kappa N\tau}\max_{u_\params\in U_\params}V\left(\ol\xi_{x_su_\params}(\tau),x_s\right)\right)\right)^{1/q},
\end{align}
where
\begin{align}\label{eta}
\eta\Let\max_{\substack{u_\params\in U_\params,x_\params\in X_\params\\x'_\params\in\mathbf{Post}_{u_\params}(x_\params)}}\left\Vert\ol\xi_{\ol{H}_\params(x_\params)u_\params}(\tau)-\ol{H}_\params\left(x'_\params\right)\right\Vert.
\end{align}
\end{lemma}

The proof of Lemma \ref{lemma1} is provided in the Appendix.
The next lemma provides similar result as the one in Lemma \ref{lemma1}, but without explicitly using any Lyapunov function.

\begin{lemma}\label{lemma2}
Consider a $\delta$-ISS-M$_q$ stochastic control system $\Sigma$ and its corresponding symbolic model $\ol{S}_\params(\Sigma)$. We have:
\begin{align}\label{upper_bound1}
\eta\leq&\left(\beta\left(\max_{u_\params\in U_\params}\left\Vert\ol\xi_{x_su_\params}(\tau)-x_s\right\Vert^q,N\tau\right)\right)^{1/q},
\end{align}
where
$\eta$ is given in \eqref{eta} and $\beta$ is the $\mathcal{KL}$ function appearing in \eqref{delta_PISS}.\qed
\end{lemma}

The proof of Lemma \ref{lemma2} is provided in the Appendix.
The next two lemmas provide similar results as Lemmas \ref{lemma1} and \ref{lemma2}, but by using the symbolic model $S_\params(\Sigma)$ with probabilistic output values rather than $\ol{S}_\params(\Sigma)$ with non-probabilistic output values.

\begin{lemma}\label{lemma4}
Consider a stochastic control system $\Sigma$, admitting a $\delta$-ISS-M$_q$ Lyapunov function $V$, and consider its corresponding symbolic model $S_\params(\Sigma)$. One has:
\begin{align}\label{upper_bound4}
\widehat\eta\leq&\left(\ul\alpha^{-1}\left(\mathsf{e}^{-\kappa N\tau}\max_{u_\params\in U_\params}\EE\left[V\left(\xi_{x_su_\params}(\tau),x_s\right)\right]\right)\right)^{1/q},
\end{align}
where
\begin{align}\label{eta1}
\widehat\eta\Let\max_{\substack{u_\params\in U_\params,x_\params\in X_\params\\x'_\params\in\mathbf{Post}_{u_\params}(x_\params)}}\EE\left[\left\Vert\xi_{H_\params(x_\params)u_\params}(\tau)-H_\params\left(x'_\params\right)\right\Vert\right].
\end{align}
\end{lemma}

\begin{proof}
The proof is similar to the one of Lemma \ref{lemma1} and can be shown by using convexity of $\ul\alpha$ and Jensen inequality \cite{oksendal}. \qed
\end{proof}


\begin{lemma}
Consider a $\delta$-ISS-M$_q$ stochastic control system $\Sigma$ and its corresponding symbolic model $S_\params(\Sigma)$. We have:
\begin{align}\label{upper_bound5}
\widehat\eta\leq&\left(\beta\left(\max_{u_\params\in U_\params}\EE\left[\left\Vert\xi_{x_su_\params}(\tau)-x_s\right\Vert^q\right],N\tau\right)\right)^{1/q},
\end{align}
where
$\widehat\eta$ is given in \eqref{eta1} and $\beta$ is the $\mathcal{KL}$ function appearing in \eqref{delta_PISS}.\qed
\end{lemma}

\begin{proof}
The proof is similar to the one of Lemma \ref{lemma2} and can be shown by using Jensen inequality \cite{oksendal}. 
\end{proof}

\begin{remark}
It can be readily verified that by choosing $N$ sufficiently large, $\eta$ and $\widehat\eta$ can be made arbitrarily small. 
One can as well try to reduce the upper bound for $\eta$ (in \eqref{upper_bound} for example) by selecting the initial point $x_s$ as follows: 
\begin{align}\label{optimal}
x_s=\displaystyle\argmin_{x\in\R^n}\max_{u_\params\in U_\params}V\left(\ol\xi_{xu_\params}(\tau),x\right).
\end{align}
\end{remark}

We can now present the first main result of the paper,
which relates the existence of a $\delta$-ISS-M$_q$ Lyapunov function to the construction of an approximately bisimilar symbolic model.
\begin{theorem}\label{main_theorem}
Consider a stochastic control system $\Sigma$ with $f(0_n,0_m) = 0_n$ and $\sigma(0_n) = 0_{n\times{p}}$, admitting a $\delta$-ISS-M$_q$ Lyapunov function $V$, of the form of the one explained in Lemma \ref{lemma3}, such that \eqref{supplement} holds for some concave $\widehat\gamma\in\mathcal{K}_\infty$. Let $\eta$ be given by \eqref{eta}. For any $\varepsilon\in\R^+$ and any tuple $\mathsf{q}=\left(\tau,\mu,N,x_s\right)$ of parameters satisfying $\mu\leq\boxspan(\mathsf{U})$ and
\begin{align}
\label{bisim_cond}
\mathsf{e}^{-\kappa\tau}\underline\alpha\left(\varepsilon^q\right)+\frac{1}{\mathsf{e}\kappa}\rho(\mu)+\widehat\gamma\left(\left(h_{x_s}((N+1)\tau)\right)^{\frac{1}{q}}+\eta\right)&\leq\underline\alpha\left(\varepsilon^q\right),
\end{align}
the relation (cf. Definition \ref{APSR})$$R=\left\{(x_\tau,x_\params)\in X_\tau\times X_\params\,\,|\,\,\EE\left[V\left(x_\tau,\ol{H}_\params(x_\params)\right)\right]\leq\ul\alpha\left(\varepsilon^q\right)\right\}$$ is an $\varepsilon$-approximate bisimulation relation between $\ol{S}_{\mathsf{q}}(\Sigma)$ and $S_{\tau}(\Sigma)$.\qed
\end{theorem}

  The proof can be found in the Appendix. 
By choosing $N$ sufficiently large and using the results in Lemmas \ref{lemma3} and \ref{lemma1}, one can enforce $h_{x_s}((N+1)\tau)$ and $\eta$ in \eqref{bisim_cond} to be sufficiently small. Hence, it can be readily seen that for a given precision $\varepsilon$,
there always exists a sufficiently small value of $\mu$ and a large value of $N$, such that the condition in (\ref{bisim_cond}) is satisfied. 
A result similar as that in Theorem \ref{main_theorem} can be recovered for a $\delta$-ISS non-probabilistic control system $\ol\Sigma$, as provided in the following corollary.

\begin{corollary}\label{corollary1}
Consider a non-probabilistic control system $\ol\Sigma$ admitting a $\delta$-ISS Lyapunov function $V$ such that \eqref{supplement} holds for some $\widehat\gamma\in\mathcal{K}_\infty$. Let $\eta$ be given by \eqref{eta}. For any $\varepsilon\in\R^+$ and any tuple $\mathsf{q}=\left(\tau,\mu,N,x_s\right)$ of parameters satisfying $\mu\leq\boxspan(\mathsf{U})$ and
\begin{align}
\label{bisim_cond11}
\mathsf{e}^{-\kappa\tau}\underline\alpha\left(\varepsilon\right)+\frac{1}{\mathsf{e}\kappa}\rho(\mu)+\widehat\gamma\left(\eta\right)&\leq\underline\alpha\left(\varepsilon\right),
\end{align}
the relation$$R=\left\{(x_\tau,x_\params)\in X_\tau\times X_\params\,\,|\,\,V\left(x_\tau,\ol{H}_\params(x_\params)\right)\leq\ul\alpha\left(\varepsilon\right)\right\}$$ is an $\varepsilon$-approximate bisimulation relation between $\ol{S}_{\mathsf{q}}(\ol\Sigma)$ and $S_{\tau}(\ol\Sigma)$.\qed
\end{corollary}

The proof is similar to the one of Theorem \ref{main_theorem}. 
In order to mitigate the conservativeness that might result from using Lyapunov functions, the next theorem provides a result that is similar to the one of Theorem \ref{main_theorem},
which is however not obtained by explicit use of $\delta$-ISS-M$_q$ Lyapunov functions,
but by using functions $\beta$ and $\gamma$ as in (\ref{delta_PISS}).

\begin{theorem}\label{main_theorem2}
Consider a $\delta$-ISS-M$_q$ stochastic control system $\Sigma$, satisfying the result of Lemma \ref{lemma3}. Let $\eta$ be given by \eqref{eta}. For any $\varepsilon\in\R^+$, and any tuple $\mathsf{q}=\left(\tau,\mu,N,x_s\right)$ of parameters satisfying $\mu\leq\boxspan(\mathsf{U})$ and
\begin{align}\label{bisim_cond2}
\left(\beta\left(\varepsilon^q,\tau\right)+\gamma(\mu)\right)^{\frac{1}{q}}+\left(h_{x_s}((N+1)\tau)\right)^{\frac{1}{q}}+\eta\leq\varepsilon,
\end{align}
the relation$$R=\left\{(x_\tau,x_\params)\in X_\tau\times X_\params\,\,|\,\,{\left(\mathbb{E}\left[\left\Vert x_{\tau}-\ol{H}_\params(x_{\params})\right\Vert^q\right]\right)^{\frac{1}{q}}\leq\varepsilon} \right\}$$ is an $\varepsilon$-approximate bisimulation relation between $\ol{S}_{\mathsf{q}}(\Sigma)$ and $S_{\tau}(\Sigma)$.\qed
\end{theorem}

The proof can be found in the Appendix.
By choosing $N$ sufficiently large and using the results in Lemmas \ref{lemma3} and \ref{lemma2}, one can force $h_{x_s}((N+1)\tau)$ and $\eta$ in \eqref{bisim_cond2} to be sufficiently small. Hence, it can be readily seen that for a given precision $\varepsilon$,
there always exist a sufficiently large value of $\tau$ and $N$ and a small enough value of $\mu$ such that the condition in (\ref{bisim_cond2}) is satisfied.
However, unlike the result in Theorem \ref{main_theorem},
notice that here for a given fixed sampling time $\tau$,
one may not find any values of $N$ and $\mu$ satisfying (\ref{bisim_cond2}) because the quantity $\left(\beta\left(\varepsilon^q,\tau\right)\right)^{\frac{1}{q}}$ may be larger than $\varepsilon$.
The symbolic model $\ol{S}_\params(\Sigma)$, computed using the parameter $\mathsf{q}$ provided in Theorem \ref{main_theorem2} (whenever existing), is likely to have fewer states than the model computed using the parameter $\mathsf{q}$ provided in Theorem \ref{main_theorem} -- a similar fact has been experienced in the first example in \cite{majid8}. A result similar to the one in Theorem \ref{main_theorem2} can be fully recovered for a $\delta$-ISS non-probabilistic control system $\ol\Sigma$, as provided in the following corollary. 

\begin{corollary}\label{corollary2}
Consider a $\delta$-ISS non-probabilistic control system $\ol\Sigma$. Let $\eta$ be given by \eqref{eta}. For any $\varepsilon\in\R^+$, and any tuple $\mathsf{q}=\left(\tau,\mu,N,x_s\right)$ of parameters satisfying $\mu\leq\boxspan(\mathsf{U})$ and\footnote{Here, $\beta$ and $\gamma$ are the $\mathcal{KL}$ and $\mathcal{K}_\infty$ functions, respectively, appearing in \eqref{delta_ISS}.}
\begin{align}\label{bisim_cond111}
\beta\left(\varepsilon,\tau\right)+\gamma(\mu)+\eta\leq\varepsilon,
\end{align}
the relation$$R=\left\{(x_\tau,x_\params)\in X_\tau\times X_\params\,\,|\,\,\left\Vert x_{\tau}-\ol{H}_\params(x_{\params})\right\Vert\leq\varepsilon \right\}$$ is an $\varepsilon$-approximate bisimulation relation between $\ol{S}_{\mathsf{q}}(\ol\Sigma)$ and $S_{\tau}(\ol\Sigma)$.\qed
\end{corollary}

The proof is similar to the one of Theorem \ref{main_theorem2}. 
The next theorems provide results that are similar to those of Theorems \ref{main_theorem} and \ref{main_theorem2}, but by using the symbolic model $S_{\params}(\Sigma)$ with probabilistic output values rather than $\ol{S}_{\params}(\Sigma)$ with non-probabilistic output values.

\begin{theorem}\label{main_theorem3}
Consider a stochastic control system $\Sigma$, admitting a $\delta$-ISS-M$_q$ Lyapunov function $V$ such that \eqref{supplement} holds for some concave $\widehat\gamma\in\mathcal{K}_\infty$. Let $\widehat\eta$ be given by \eqref{eta1}. For any $\varepsilon\in\R^+$ and any tuple $\mathsf{q}=\left(\tau,\mu,N,x_s\right)$ of parameters satisfying $\mu\leq\boxspan(\mathsf{U})$ and
\begin{align}
\label{bisim_cond3}
\mathsf{e}^{-\kappa\tau}\underline\alpha\left(\varepsilon^q\right)+\frac{1}{\mathsf{e}\kappa}\rho(\mu)+\widehat\gamma\left(\widehat\eta\right)&\leq\underline\alpha\left(\varepsilon^q\right),
\end{align}
the relation$$R=\left\{(x_\tau,x_\params)\in X_\tau\times X_\params\,\,|\,\,\EE\left[V(x_\tau,H_\params(x_\params))\right]\leq\ul\alpha\left(\varepsilon^q\right)\right\}$$ is an $\varepsilon$-approximate bisimulation relation between ${S}_{\mathsf{q}}(\Sigma)$ and $S_{\tau}(\Sigma)$.\qed
\end{theorem}

The proof is similar to the one of Theorem \ref{main_theorem}.

\begin{theorem}\label{main_theorem4}
Consider a $\delta$-ISS-M$_q$ stochastic control system $\Sigma$. Let $\widehat\eta$ be given by \eqref{eta1}. For any $\varepsilon\in\R^+$, and any tuple $\mathsf{q}=\left(\tau,\mu,N,x_s\right)$ of parameters satisfying $\mu\leq\boxspan(\mathsf{U})$ and
\begin{align}\label{bisim_cond4}
\left(\beta\left(\varepsilon^q,\tau\right)+\gamma(\mu)\right)^{\frac{1}{q}}+\widehat\eta\leq\varepsilon,
\end{align}
the relation$$R=\left\{(x_\tau,x_\params)\in X_\tau\times X_\params\,\,|\,\,{\left(\mathbb{E}\left[\left\Vert x_{\tau}-H_\params(x_{\params})\right\Vert^q\right]\right)^{\frac{1}{q}}\leq\varepsilon} \right\}$$ is an $\varepsilon$-approximate bisimulation relation between ${S}_{\mathsf{q}}(\Sigma)$ and $S_{\tau}(\Sigma)$.\qed
\end{theorem}

The proof is similar to the one of Theorem \ref{main_theorem2}.

\begin{remark}\label{remark1}
The symbolic model $S_\params(\Sigma)$, computed using the parameter $\mathsf{q}$ provided in Theorem \ref{main_theorem3} (resp. Theorem \ref{main_theorem4}), has fewer (or at most equal number of) states than the symbolic model $\ol{S}_\params(\Sigma)$, computed by using the parameter $\params$ provided in Theorem \ref{main_theorem} (resp. Theorem \ref{main_theorem2}) while having the same precision. However, the symbolic model $S_\params(\Sigma)$ has states with probabilistic output values, rather than non-probabilistic ones, which is likely to require more involved control synthesis procedures (cf. Subsection \ref{ssec:ltl.control}). 
\qed 
\end{remark}

\begin{remark}\label{remark5}
Although we assume that the set $\mathsf{U}$ is infinite,
Theorems \ref{main_theorem}, \ref{main_theorem2}, \ref{main_theorem3}, and \ref{main_theorem4} and Corollaries \ref{corollary1} and \ref{corollary2} still hold when the set $\mathsf{U}$ is finite,
with the following modifications.
First, the systems $\Sigma$ and $\ol\Sigma$ are required to satisfy the properties (\ref{delta_PISS}) and \eqref{delta_ISS}, respectively, for $\upsilon=\upsilon'$.
Second, take $U_\params=\mathsf{U}$
in the definitions of $\ol{S}_\params(\Sigma)$ (resp. ${S}_{\params}(\Sigma)$) and $\ol{S}_\params(\ol\Sigma)$.
Finally, in the conditions (\ref{bisim_cond}), \eqref{bisim_cond11}, (\ref{bisim_cond2}), \eqref{bisim_cond111}, \eqref{bisim_cond3}, and \eqref{bisim_cond4} set $\mu=0$. \qed
\end{remark}

Finally, we establish the results on the existence of symbolic model $\ol{S}_\params(\Sigma)$ (resp. $S_\params(\Sigma)$) such that \mbox{$\ol{S}_\params(\Sigma)\cong_{\mathcal{S}}^{\varepsilon}S_\tau(\Sigma)$} (resp. \mbox{$S_\params(\Sigma)\cong_{\mathcal{S}}^{\varepsilon}S_\tau(\Sigma)$}) and \mbox{$\ol{S}_\params(\ol\Sigma)\cong_{\mathcal{S}}^{\varepsilon}S_\tau(\ol\Sigma)$}.

\begin{theorem}\label{main_theorem5}
Consider the results in Theorem \ref{main_theorem}. If we select $$X_{\tau0}=\left\{x\in\R^n|\left\Vert x-\ol{H}_\params(x_{\params0})\right\Vert\leq\left(\ol\alpha^{-1}\left(\ul\alpha\left(\varepsilon^q\right)\right)\right)^{\frac{1}{q}},\exists x_{\params0}\in X_{\params0}\right\},$$then we have \mbox{$\ol{S}_\params(\Sigma)\cong_{\mathcal{S}}^{\varepsilon}S_\tau(\Sigma)$}.\qed
\end{theorem}

\begin{proof}
We start by proving that \mbox{$S_{\tau}(\Sigma)\preceq^{\varepsilon}_\mathcal{S}\ol{S}_{\params}(\Sigma)$}. For every $x_{\tau 0}\in{X_{\tau 0}}$ there always exists \mbox{$x_{\params 0}\in{X}_{\params 0}$} such that $\Vert{x_{\tau0}}-\ol{H}_\params(x_{\params0})\Vert\leq\left(\ol\alpha^{-1}\left(\ul\alpha\left(\varepsilon^q\right)\right)\right)^{\frac{1}{q}}$. Then,
\begin{align}\nonumber
\mathbb{E}\left[V\left({x_{\tau0}},\ol{H}_\params(x_{\params0})\right)\right]&=V\left({x_{\tau0}},\ol{H}_\params(x_{\params0})\right)\leq\overline\alpha(\Vert x_{\tau0}-\ol{H}_\params(x_{\params0})\Vert^q)\leq\underline\alpha\left(\varepsilon^q\right),
\end{align}
since $\overline\alpha$ is a $\mathcal{K}_\infty$ function.
Hence, \mbox{$\left(x_{\tau0},x_{\params0}\right)\in{R}$} implying that \mbox{$S_{\tau}(\Sigma)\preceq^{\varepsilon}_\mathcal{S}\ol{S}_{\params}(\Sigma)$}. In a similar way, we can show that \mbox{$\ol{S}_{\params}(\Sigma)\preceq^{\varepsilon}_{\mathcal{S}}S_{\tau}(\Sigma)$} which completes the proof. 
\end{proof}

The next theorem provides a similar result in line with the one of previous theorem, but by using a different relation.

\begin{theorem}\label{main_theorem6}
Consider the results in Theorem \ref{main_theorem2}. If we select $$X_{\tau0}=\left\{x\in\R^n\,\,|\,\,\left\Vert x-\ol{H}_\params(x_{\params0})\right\Vert\leq\varepsilon,~\exists x_{\params0}\in X_{\params0}\right\},$$then we have \mbox{$\ol{S}_\params(\Sigma)\cong_{\mathcal{S}}^{\varepsilon}S_\tau(\Sigma)$}.\qed
\end{theorem}

\begin{proof}
We start by proving that \mbox{$S_{\tau}(\Sigma)\preceq^{\varepsilon}_\mathcal{S}\ol{S}_{\params}(\Sigma)$}. For every $x_{\tau 0}\in{X_{\tau 0}}$ there always exists \mbox{$x_{\params 0}\in{X}_{\params 0}$} such that $\Vert{x_{\tau0}}-\ol{H}_\params(x_{\params0})\Vert\leq\varepsilon$ and $\left(\EE\left[\left\Vert x_{\tau0}-\ol{H}_\params(x_{\params0})\right\Vert^q\right]\right)^{\frac{1}{q}}\leq\varepsilon$.
Hence, \mbox{$\left(x_{\tau0},x_{\params0}\right)\in{R}$} implying that \mbox{$S_{\tau}(\Sigma)\preceq^{\varepsilon}_\mathcal{S}\ol{S}_{\params}(\Sigma)$}. In a similar way, we can show that \mbox{$\ol{S}_{\params}(\Sigma)\preceq^{\varepsilon}_{\mathcal{S}}S_{\tau}(\Sigma)$} which completes the proof. 
\end{proof}

The next two corollaries provide similar results as the ones of Theorems \ref{main_theorem5} and \ref{main_theorem6}, but for non-probabilistic control systems $\ol\Sigma$.

\begin{corollary}
Consider the results in Corollary \ref{corollary1}. If we select $$X_{\tau0}=\left\{x\in\R^n|\left\Vert x-\ol{H}_\params(x_{\params0})\right\Vert\leq\left(\ol\alpha^{-1}\left(\ul\alpha\left(\varepsilon\right)\right)\right),\exists x_{\params0}\in X_{\params0}\right\},$$then we have \mbox{$\ol{S}_\params(\ol\Sigma)\cong_{\mathcal{S}}^{\varepsilon}S_\tau(\ol\Sigma)$}.\qed
\end{corollary}

The proof is similar to the one of Theorem \ref{main_theorem5}.

\begin{corollary}
Consider the results in Corollary \ref{corollary2}. If we select $$X_{\tau0}=\left\{x\in\R^n\,\,|\,\,\left\Vert x-\ol{H}_\params(x_{\params0})\right\Vert\leq\varepsilon,~\exists x_{\params0}\in X_{\params0}\right\},$$then we have \mbox{$\ol{S}_\params(\ol\Sigma)\cong_{\mathcal{S}}^{\varepsilon}S_\tau(\ol\Sigma)$}.\qed
\end{corollary}

The proof is similar to the one of Theorem \ref{main_theorem6}.
The next two theorems provide similar results as the ones of Theorems \ref{main_theorem5} and \ref{main_theorem6}, but by using the symbolic model $S_\params(\Sigma)$.

\begin{theorem}
Consider the results in Theorem \ref{main_theorem3}. Let $\mathcal{A}$ denote the set of all $\R^n$-valued random variables, measurable over $\sigalg_0$. If we select
\begin{align}\notag
&X_{\tau0}=\left\{a\in \mathcal{A}|\left(\EE\left[\left\Vert a-H_\params(x_{\params0})\right\Vert^q\right]\right)^{\frac{1}{q}}\leq\left(\ol\alpha^{-1}\left(\ul\alpha\left(\varepsilon^q\right)\right)\right)^{\frac{1}{q}},\exists x_{\params0}\in X_{\params0}\right\},
\end{align}
then we have \mbox{$S_\params(\Sigma)\cong_{\mathcal{S}}^{\varepsilon}S_\tau(\Sigma)$}.\qed
\end{theorem}

The proof is similar to the one of Theorem \ref{main_theorem5}.

\begin{theorem}
Consider the results in Theorem \ref{main_theorem4}. Let $\mathcal{A}$ denote the set of all $\R^n$-valued random variables, measurable over $\sigalg_0$. If we select
\begin{align}\notag
X_{\tau0}=\left\{a\in \mathcal{A}\,\,|\,\,\left(\EE\left[\left\Vert a-H_\params(x_{\params0})\right\Vert^q\right]\right)^{\frac{1}{q}}\leq\varepsilon,~\exists x_{\params0}\in X_{\params0}\right\},
\end{align}
then we have \mbox{$S_\params(\Sigma)\cong_{\mathcal{S}}^{\varepsilon}S_\tau(\Sigma)$}.\qed
\end{theorem}

The proof is similar to the one of Theorem \ref{main_theorem6}.

\subsection{Control synthesis over $S_\params(\Sigma)$}
\label{ssec:ltl.control}
Note that both $\ol{S}_\q(\Sigma)$ and $S_\params(\Sigma)$ are finite systems.
The only difference is that the outputs of the former system are always non-probabilistic points,
whereas those of the latter can be non-degenerate random variables.
Let us describe the control synthesis for these systems over quantitative specifications, 
and for example over the safety formula $\square \varphi_A$, for $A\subset \R^n\subset Y$ 
(as already been used in Subsection \ref{ssec:fin.abstr}).
Clearly, since the original system $S_\tau(\Sigma)$ is stochastic in the sense that its outputs are non-degenerate random variables similarly to $S_\params(\Sigma)$,
it would be too conservative to require that it satisfies the formula exactly.
Thus, we are rather interested in an input policy that makes $S_\tau(\Sigma)$ satisfy $\square \varphi_{A_\ve}$ with some $\ve>0$:
recall from Subsection \ref{ssec:fin.abstr} that the latter LTL formula can be satisfied by non-degenerate random variables,
in contrast to $\square \varphi_A$.
Let us recap how to use abstractions for this task,
and let us start with $\ol{S}_\q(\Sigma)$ belonging to a more familiar type of systems whose outputs are non-probabilistic.

We label a state $x_\q$ of $\ol{S}_\q(\Sigma)$ with $A$ if $\ol{H}_\q(x_\q)\in A$ and,
say, with $B$ otherwise.
As a result, we obtain a transition system with labels over the states and can synthesize a control strategy by solving a safety game \cite{paulo} that makes an output run of $\ol{S}_\params(\Sigma)$ satisfy $\square \varphi_A$.
After that,
we can exploit $\ve$-approximate bisimilarity to guarantee that the refined input policy makes the corresponding output run of the original system satisfy $\square \varphi_{A_\ve}$.

The main subtlety in the case of $S_\params(\Sigma)$ is how to label its states.
We cannot do this
as for $\ol{S}_\q(\Sigma)$,
since $H_\q(x_\q)$ may never be an element of $A$ for any $x_\q\in X_\q$:
indeed, the latter is a set of non-probabilistic points,
whereas all the outputs of $S_\params(\Sigma)$ can happen to be non-degenerate random variables.
In order to cope with this issue,
we propose to relax the original problem and at the same time to strengthen the quality of the abstraction. Namely,
we can consider a relaxed problem $\square \varphi_{A_{\delta}}$ over the abstraction $ S_\params(\Sigma)$, for some $\delta\in]0~\varepsilon[$,
where the latter is now required to be $(\varepsilon-\delta)$-approximate (rather than just $\ve$-approximate) bisimilar to the original system.
Clearly $\left(A_{\delta}\right)_{\varepsilon-\delta}\subseteq A_\ve$,
so that whenever the control policy for $\square \varphi_{A_{\delta}}$ is synthesized over $ S_\params(\Sigma)$,
its refined version is guaranteed to enforce $\square \varphi_{A_\ve}$ over the original system.
Thanks to the fact that $A_{\delta}$ contains non-degenerate random variables,
we eliminate the conservativeness presented before in the sense that it is likely that there are now points $x_\q\in X_\q$ in $ S_\params(\Sigma)$ such that $H_\q(x_\q)\in A_{\delta}$.
The only remaining question is how to check whether $H_\q(x_\q)\in A_{\delta}$.
To answer this question,
we check that the distance
\begin{equation}\label{eq:dist.rv.set}
  \mathbf{d}\left(H_\q(x_\q),A\right) = \inf_{a\in A}\left(\EE\|\xi_{x_sx_\q}(N\tau) - a\|^q\right)^{1/q}
\end{equation}
is smaller than $\delta$, which involves both computing the expectation over the solution of the SDE,
and optimizing the value of this expectation.
Clearly, such a computation in general cannot be done analytically,
and the evaluation of the expectation itself is a highly non-trivial task unless the SDE has a very special form.

We propose a Monte Carlo approach to compute an approximation of the quantity in \eqref{eq:dist.rv.set} by means of empirical expectations.
Using such an approach,
we can estimate $\mathbf{d}\left(H_\q(x_\q),A\right)$ only up to some precision,
say $\theta$.
If the estimated distance is less than $\delta - \theta$,
we are safe to label $x_\q$ with $A$,
whereas all other states are labeled by $B$.
Furthermore,
since this result is based on a Monte Carlo method,
it holds true only with a certain confidence level $1-\pi$ where $\pi \in [0~1]$.
The benefit of our approach is that it is not only valid asymptotically (as the number of samples grows to infinity),
but we are also able to provide a number of simulations that is sufficient to estimate $\mathbf{d}\left(H_\q(x_\q),A\right)$ with any given precision $\theta$ and with any given confidence $1-\pi$.
This can be considered as an extension of the well-known Hoeffding's inequality \cite{hoeffding} to the case when one has to deal with an optimization problem.
Note that regardless of the specification of interest,
the main task over $S_\params(\Sigma)$ is always to compute some distance as in \eqref{eq:dist.rv.set} for any set that appears in the specification,
so the method below applies not only to the safety formula $\square \varphi_A$,
but also to more general formulae,
which are left as object of the future research.

Suppose that $A$ as in \eqref{eq:dist.rv.set} is a compact subset of $\R^n$,
and let $A^r$ be the smallest subset of $\left[\R^n\right]_r$ such that $A\subseteq\bigcup_{p\in A^r}\mathcal{B}_{\frac{r}{2}}({p})$.
Let $M$ be the number of samples and let
\begin{equation*}
  \mathbf{d}^r_M := \min_{a\in A^r}\left(\frac1M\sum_{i=1}^M \left\|\xi^i_{x_sx_\q}(N\tau)-a\right\|^q\right)^{\frac1q},
\end{equation*}
where the superscript $i$ denotes the index of samples. Now, we have the following theorem.

\begin{theorem}
\label{thm:number.samples}
  For any stochastic control system $\Sigma$ one has $|\mathbf{d}\left(H_\q(x_\q),A\right) - \mathbf{d}^r_M|\leq \theta$ with confidence of at least $1-\pi$, 
  given that $r < 2\theta$ and that
  \begin{equation*}
    M \geq \frac{|A^r|b(a^*, 2q)}{\pi(\theta - r/2)^{2q}},
  \end{equation*}
  where $b(a, p) := (1 + |x_s - a|^p)\mathrm e^{p(p+1)\max\{L_x, Z\}N\tau}$ and $a^* \in \argmax_{a'\in A^r}\|x_s - a'\|$.
\end{theorem}

The proof can be found in the Appendix.
Let us make some comments on Theorem \ref{thm:number.samples}.
First of all,
no matter how many distances one has to evaluate,
one can always use the same samples $\xi^i$ and there is no need to generate new samples.
Second, to the best of our knowledge,
logarithmic bounds on $M$ (as per \cite{kloeden}) are not available in this general case due to the fact that we deal with an unbounded state space.

\subsection{Relationship with existing results in the literature}
Note that given any precision $\varepsilon$ and sampling time $\tau$, one can always use the results in Theorem \ref{main_theorem5} to construct a symbolic model $\ol{S}_\params(\Sigma)$ that is $\varepsilon$-approximate bisimilar to $S_\tau(\Sigma)$ without any state set discretization. Note that the results in Theorem 5.1 in \cite{majid8} also provide symbolic models that are $\varepsilon$-approximate bisimilar to $S_\tau(\Sigma)$. However, the results in \cite{majid8} require state set discretization and cannot be applied for any sampling time $\tau$ if the precision $\varepsilon$ is lower than the thresholds introduced in inequality (5.5) in \cite{majid8}. Furthermore, while the results in \cite{majid8} only provide symbolic models with non-probabilistic output values, the ones in this work provide symbolic models with probabilistic output values as well, which can result in less conservative symbolic models (cf. Remark \ref{remark1}).

One can compare the results provided in Theorems \ref{main_theorem} (corr. \ref{main_theorem5}) and \ref{main_theorem2} (corr. \ref{main_theorem6}) with the results provided in Theorems 5.1 and 5.3 in \cite{majid8} in terms of the size of the generated symbolic models. One can readily verify that the precisions of the symbolic models proposed here and the ones proposed in \cite{majid8} are approximately the same as long as both use the same input set quantization parameter $\mu$ and the state space quantization parameter, called $\nu$, in \cite{majid8} is equal to the parameter $\eta$ in \eqref{eta}, i.e. $\nu\leq\left(\ul\alpha^{-1}\left(\mathsf{e}^{-\kappa N\tau}\eta_0\right)\right)^{1/q}$, where $\eta_0=\max_{u_\params\in U_\params}V\left(\ol\xi_{x_su_\params}(\tau),x_s\right)$. The reason their precisions are approximately (rather than exactly) the same is because we use $h_{x_s}\left(\sigma,(N+1)\tau\right)$ in conditions \eqref{bisim_cond} and \eqref{bisim_cond2} in this paper rather than $h(\tau)=\sup_{x\in\mathsf{D}}h_x(\tau)$ that is being used in conditions 5.4 and 5.14 in \cite{majid8} for a compact set $\mathsf{D}\subset\R^n$. By assuming that $h_{x_s}\left(\sigma,(N+1)\tau\right)^{\frac{1}{q}}$ and $h(\tau)^{\frac{1}{q}}$ are much smaller than $\eta$ and $\nu$, respectively, or $h_{x_s}\left(\sigma,(N+1)\tau\right)\approx h(\tau)$, one should expect to obtain the same precisions for the symbolic models provided here and those provided in \cite{majid8} under the aforementioned conditions.

The number of states of the proposed symbolic model in this paper is $\left\vert\left[\mathsf{U}\right]_\mu\right\vert^N$. Assume that we are interested in the dynamics of $\Sigma$ on a compact set $\mathsf{D}\subset\R^n$. Since the set of states of the proposed symbolic model in \cite{majid8} is $\left[\mathsf{D}\right]_{\nu}$, its size is $\left\vert\left[\mathsf{D}\right]_{\nu}\right\vert=\frac{K}{\nu^n}$, where $K$ is a positive constant proportional to the volume of $\mathsf{D}$. Hence, it is more convenient to use the proposed symbolic model here rather than the one proposed in \cite{majid8} as long as:
\begin{align}\nonumber
\left\vert\left[\mathsf{U}\right]_\mu\right\vert^N\leq\frac{K}{\left(\ul\alpha^{-1}\left(\mathsf{e}^{-\kappa N\tau}\eta_0\right)\right)^{n/q}}.
\end{align}
Without loss of generality, one can assume that $\ul\alpha({r})=r$ for any $r\in\R_0^+$. Hence, for sufficiently large value of $N$, it is more convenient to use the proposed symbolic model here in comparison with the one proposed in \cite{majid8} as long as:
\begin{align}\label{criterion}
\left\vert\left[\mathsf{U}\right]_\mu\right\vert\mathsf{e}^{\frac{-\kappa\tau n}{q}}\leq1.
\end{align}

Note that the methodology proposed in this paper allows us to construct less conservative symbolic models with probabilistic output values while the proposed one in \cite{majid8} only provides conservative symbolic models with non-probabilistic output values.

\section{Example}
We show the effectiveness of the results presented in this work by constructing a bisimilar symbolic model for the model of a road network, 
which is divided in 5 cells of 250 meters with 2 entries and 2 ways out, as depicted schematically in Figure \ref{traffic}. 
The model is borrowed from \cite{corronc}, however it is now affected by noise and newly described in continuous time.

\begin{figure}[h]
\begin{center}
\includegraphics[width=14cm]{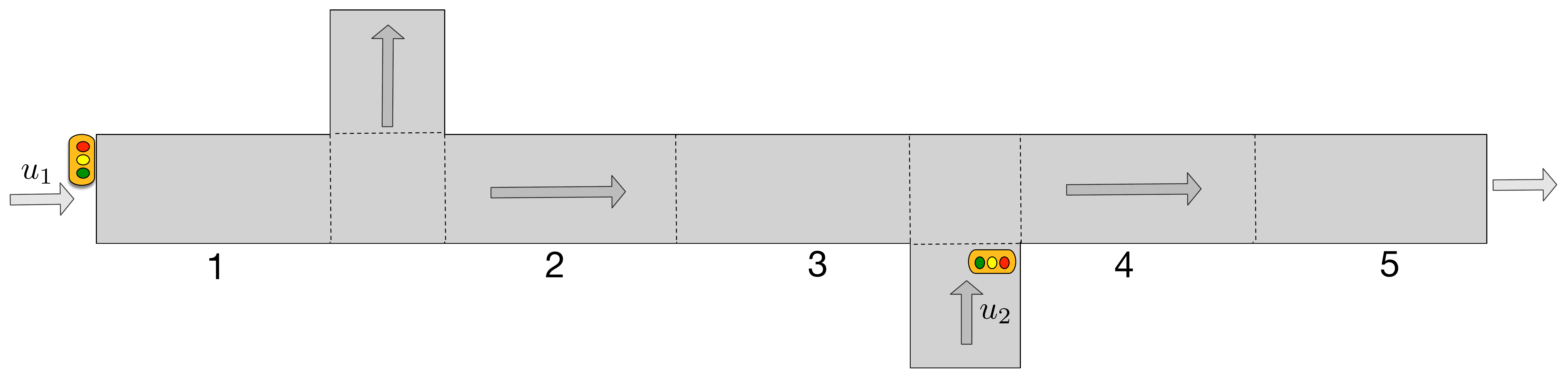}
\end{center}
\caption{Model of a road divided in 5 cells with 2 entries and 2 ways out.}
\label{traffic}
\end{figure}

The two entries are controlled by
traffic lights, denoted by $u_1$ and $u_2$, that enable (green light) or
not (red light) the vehicles to pass. In this model the length of a cell is in kilometres (0.25 km), and the flow
speed of the vehicles is 70 kilometres per hour (km/h). Moreover, during the sampling time interval $\tau$, it is assumed that 6 vehicles pass the entry controlled by the light
$u_1$, 8 vehicles pass the entry controlled by the light $u_2$, and
one quarter of vehicles that leave cell 1 goes out on the first
exit. We assume that both lights cannot be red at the same time. The model of $\Sigma$ is described by:
\begin{align}
\Sigma:\left\{\diff\xi=\left(A\xi+B\upsilon\right)\diff{t}+\xi\diff{W}_t,\right.
\end{align}
where
\begin{align}\nonumber
&A=10^4\times\left[ {\begin{array}{ccccc}
-0.0541  &  0  & 0  &  0  &  0 \\
0.3224 & -0.1370  & 0  & 0  &  0 \\
-0.7636  & 0.3224 & -0.0541  &  0 &   0\\
2.1122   & -0.7636 &   0.1260 & -0.0541  & 0\\
-6.2132  &  2.1122 &  -0.2205  &  0.1260  & -0.0541\\
 \end{array}}\right],\\\nonumber
 &B=10^4\times\left[ {\begin{array}{ccccc} 0.0696  & 0 &   0 &   0  &  0\\
 -0.2743  &  0.1402  & 0 &  0 &  0\\
 0.7075 &  -0.2743  &  0.0696  & 0  & 0\\
 -2.0081  &  0.7075 &  -0.0924  &  0.0696  &  0\\
 5.9802  & -2.0081 &   0.1911  & -0.0924  &  0.0696\\
 \end{array}}\right],
\end{align}
$\mathsf{U}=\{\mathsf{u}_0,\mathsf{u}_1,\mathsf{u}_2\}=\{[6~0~8~0~0]^T,[6~0~0~0~0]^T,[0~0~8~0~0]^T\}$, and $\xi_i$ is the number of vehicles in cell $i$ of the road. Note that $\mathcal{U}_\tau$ contains curves taking values in $\mathsf{U}$. Since $\mathsf{U}$ is finite, as explained in Remark \ref{remark5}, $\mu=0$ is to be used in (\ref{bisim_cond}), (\ref{bisim_cond2}), \eqref{bisim_cond3}, and \eqref{bisim_cond4}. One can readily verify that the function $V(x,x')=(x-x')^TP(x-x')$, for any $x,x'\in\R^5$, where 
\begin{align}\nonumber
&P=10^4\times\left[ {\begin{array}{ccccc}
76763.4393 & -2101.1583 &  3790.9182 & -155.6576 & -125.9871 \\
-2101.1583  & 10676.9437 &  1237.3552 & -86.6855 & 100.5718 \\
3790.9182 &  1237.3552 &  1823.02431 & 171.1549 & -71.1162\\
-155.6576 & -86.6855 &  171.1549 & 229.2134 & -5.5649\\
-125.9871 &  100.5718 & -71.1162 & -5.5649  & 33.3977\\
 \end{array}}\right],
\end{align}
satisfies conditions (i)-(iii) in Definition \ref{delta_PISS_Lya} with $q=2$, $\kappa=300$, $\ul\alpha({r})=1/2\lambda_{\min}(P)r$, $\ol\alpha({r})=1/2\lambda_{\max}(P)r$, $\rho({r})=5\Vert{B}\Vert^2\Vert P\Vert/(2\kappa) r^2$, $\forall r\in\R_0^+$. Hence, $\Sigma$ is $\delta$-ISS-M$_2$, equipped with the $\delta$-ISS-M$_2$ Lyapunov function $V$. Using the results of Theorem \ref{the_Lya}, provided in \cite{majid8}, one gets that functions $\beta(r,s)=\ul\alpha^{-1}\left(\ol\alpha(r)\mathsf{e}^{-\kappa{s}}\right)$ and $\gamma({r})=\ul\alpha^{-1}\left(\frac{1}{\mathsf{e}\kappa}\rho(r)\right)$ satisfy property \eqref{delta_PISS} for $\Sigma$. We choose the source state as the one proposed in \cite{corronc}, i.e. $x_s=[3.8570~~3.3750~~3.3750~~8.5177~~8.5177]^T$.

For a given precision $\varepsilon=0.5$ and fixed sampling time $\tau=0.00277$ h (10 sec), the parameter $N$ for $\ol{S}_\params(\Sigma)$, based on inequality (\ref{bisim_cond}) in Theorem \ref{main_theorem}, is obtained as 14. Therefore, the resulting cardinality of the set of states for $\ol{S}_\params(\Sigma)$ is $\left\vert\mathsf{U}\right\vert^{14}=3^{14}=4782969$. Using the aforementioned parameters, one gets $\eta\leq6.0776\times10^{-6}$, where $\eta$ is given in \eqref{eta}. Note that the results in Theorems \ref{main_theorem2} and \ref{main_theorem4} cannot be applied here because $(\beta(\varepsilon^q,\tau))^{\frac{1}{q}}>\varepsilon$. Using criterion (\ref{criterion}), one has $\left\vert\mathsf{U}\right\vert\mathsf{e}^{\frac{-\kappa\tau n}{q}}=0.37$, implying that the approach proposed in this paper is more appropriate in terms of the size of the abstraction than the one proposed in \cite{majid8}. We elaborate more on this at the end of the section.

\begin{remark}
By considering the non-probabilistic control system $\ol\Sigma$ and using the results in Corollary \ref{corollary1} and the same parameters $\params$ as the ones in $\ol{S}_\params(\Sigma)$, one obtains $\varepsilon=0.01$ in \eqref{bisim_cond11}. Therefore, as expected, $\ol S_\params(\ol\Sigma)$ (i.e. symbolic model for the non-probabilistic control system $\ol\Sigma$) provides much smaller precision than $\ol{S}_\params(\Sigma)$ (i.e. symbolic model for the stochastic control system $\Sigma$) while having the same size as $\ol{S}_\params(\Sigma)$.
\end{remark}

Now the objective, as inspired by the one suggested in \cite{corronc}, is to design a schedule for the coordination of
traffic lights enforcing $\Sigma$ to satisfy a \emph{safety} and a \emph{fairness} property. The safety part is to keep the density of traffic lower
than 16 vehicles per cell which can be encoded via the LTL specification\footnote{Note that the semantics of LTL are defined over the output behaviors of $S_{\params}(\Sigma)$.} $\Box \varphi_W$, where $W=[0~16]^5$. The fairness
part requires to alternate the accesses between the two traffic lights and to allow only 3 identical consecutive modes
of red light ensuring fairness between two traffic lights. Starting from the initial condition $x_0=[1.417  ~~ 4.993 ~~  10.962  ~~  9.791  ~~ 14.734]^T$, we obtain a periodic schedule $\upsilon=(\mathsf{u}_0\mathsf{u}_0\mathsf{u}_0\mathsf{u}_2\mathsf{u}_1\mathsf{u}_0\mathsf{u}_0\mathsf{u}_2\mathsf{u}_1\mathsf{u}_0\mathsf{u}_0\mathsf{u}_2\mathsf{u}_1\mathsf{u}_2)^\omega$ keeping $\mathsf{u}_0$ as much as possible in each period in order to maximize number of vehicles accessing the
road.  

Figure \ref{fig1} displays a few realizations of the closed-loop solution process $\xi_{x_0\upsilon}$. In Figure \ref{fig1} bottom right, we show the average value (over
100000 experiments) of the distance (in the 2nd moment metric) in time of the solution process $\xi_{x_0\upsilon}$ to the set $W$, namely $\left\Vert\xi_{x_0\upsilon}(t)\right\Vert_{W}$, where the point-to-set distance is defined as $\Vert x\Vert_W=\inf_{w\in W}\Vert x-w\Vert$. Notice that the empirical average distance is as expected lower
than the precision $\varepsilon=0.5$.

\begin{figure}[h]
\hspace{-1.4cm}
\includegraphics[width=16cm]{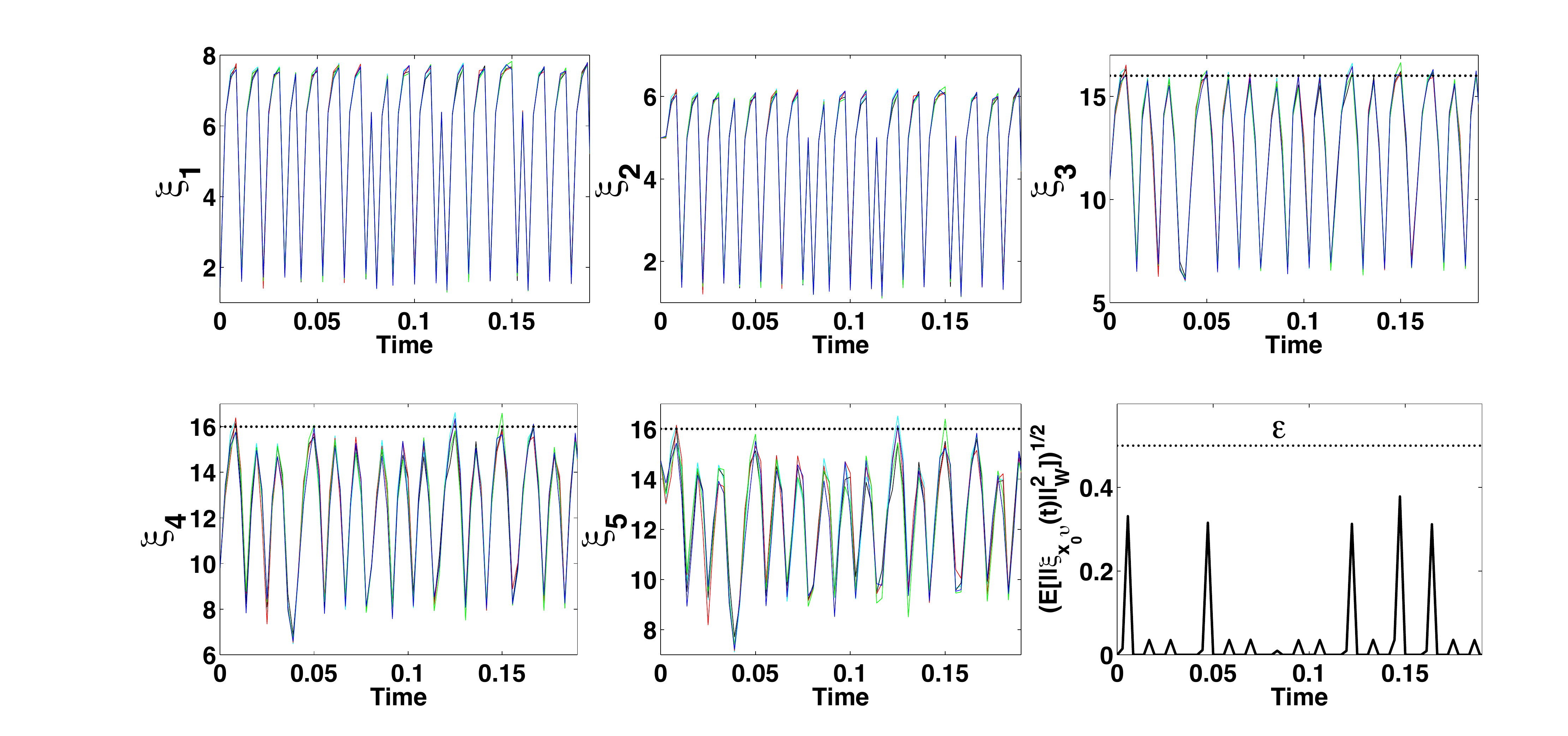}
\caption{A few realizations of the closed-loop solution process $\xi_{x_0\upsilon}$ (top panel and the first two figures from the left in the bottom panel) and the average values (over 100000 experiments) of the distance of the solution process $\xi_{x_0\upsilon}$ to the set $W$ (bottom right panel).}
\label{fig1}
\end{figure}


To compute exactly the size of the symbolic model, proposed in Theorem 5.1 in \cite{majid8}, we consider the dynamics of $\Sigma$ over the subset $\mathsf{D}=[0~16]^5$ of $\R^5$. Note that Theorem 5.3 in \cite{majid8} cannot be applied here because $(\beta(\varepsilon^q,\tau))^{\frac{1}{q}}>\varepsilon$. Using the same precision $\varepsilon=0.5$ and sampling time $\tau=0.00277$ as the ones here, and the inequalities (5.3) and (5.4) in \cite{majid8}, we obtain the state space quantization parameter as $\nu\leq6.0776\times10^{-6}$. Therefore, if one uses $\nu=6.0776\times10^{-6}$, the cardinality of the state set of the symbolic model, provided by the results in Theorem 5.1 in \cite{majid8}, is equal to $\left(\frac{16}{\nu}\right)^5=1.2645\times10^{32}$ which is much higher than the one proposed here, amounting instead to $4782969$ states.

\section{Conclusions}
In this paper we have proposed a symbolic abstraction technique for incrementally stable stochastic control systems (and corresponding non-probabilistic model), 
which features only the discretization of the input set.  
The proposed approach is potentially more scalable than the one proposed in \cite{majid8} for higher dimensional stochastic control systems. 

Future work will concentrate on efficient implementations of the symbolic models proposed in this work using Binary Decision Diagrams, 
on top of the recently developed synthesis toolbox \textsf{SCOTS} \cite{scots}, 
as well as on more efficient controller synthesis techniques.


\bibliographystyle{alpha}
\bibliography{reference}

\section*{Appendix}

\begin{proof}{of Lemma \ref{lemma1}:}
Let $x_\params\in X_\params$, where $x_\params=\left(u_1,u_2,\ldots,u_N\right)$, and $u_\params\in U_\params$. Using the definition of $\ol{S}_\params(\Sigma)$, one obtains $x'_\params=\left(u_2,\ldots,u_N,u_\params\right)\in\mathbf{Post}_{u_\params}(x_\params)$. Since $V$ is a $\delta$-ISS-M$_q$ Lyapunov function for $\Sigma$, we have:
\begin{align}\notag
\ul\alpha&\left(\left\Vert\ol\xi_{\ol{H}_\params(x_\params)u_\params}(\tau)-\ol{H}_\params\left(x'_\params\right)\right\Vert^q\right)\leq V(\ol\xi_{\ol{H}_\params(x_\params)u_\params}(\tau),\ol{H}_\params\left(x'_\params\right))\\\notag&=V(\ol\xi_{\ol\xi_{x_sx_\params}(N\tau)u_\params}(\tau),\ol\xi_{x_sx'_\params}(N\tau))=V(\ol\xi_{\ol\xi_{x_su_1}(\tau)(u_2,\ldots,u_N,u_\params)}(N\tau),\ol\xi_{x_s(u_2,\ldots,u_N,u_\params)}(N\tau))\\\label{gronwall}
&\leq\mathsf{e}^{-\kappa N\tau}V(\ol\xi_{x_su_1}(\tau),x_s).
\end{align}
We refer the interested readers to the proof of Theorem \ref{the_Lya} in \cite{majid8} to see how we derived the inequality \eqref{gronwall}.
Hence, one gets
\begin{align}\label{upper_bound2}
\Vert\ol\xi_{\ol{H}_\params(x_\params)u_\params}(\tau)-\ol{H}_\params\left(x'_\params\right)\Vert\leq(\ul\alpha^{-1}(\mathsf{e}^{-\kappa N\tau}V(\ol\xi_{x_su_1}(\tau),x_s)))^{1/q},
\end{align}
because of $\ul\alpha\in\mathcal{K}_\infty$. Since the inequality \eqref{upper_bound2} holds for all $x_\params\in X_\params$ and $u_\params\in U_\params$, and $\ul\alpha\in\mathcal{K}_\infty$, inequality \eqref{upper_bound} holds.
\end{proof}

\begin{proof}{of Lemma \ref{lemma2}:}
Let $x_\params\in X_\params$, where $x_\params=\left(u_1,u_2,\ldots,u_N\right)$, and $u_\params\in U_\params$. Using the definition of $\ol{S}_\params(\Sigma)$, one obtains $x'_\params=\left(u_2,\ldots,u_N,u_\params\right)\in\mathbf{Post}_{u_\params}(x_\params)$. Since $\Sigma$ is $\delta$-ISS-M$_q$ and using inequality \eqref{delta_PISS}, we have:
\begin{align}\notag
\Vert\ol\xi_{\ol{H}_\params(x_\params)u_\params}(\tau)-\ol{H}_\params\left(x'_\params\right)\Vert^q&=\Vert\ol\xi_{\ol\xi_{x_sx_\params}(N\tau)u_\params}(\tau)-\ol\xi_{x_sx'_\params}(N\tau)\Vert^q\\\notag
&=\Vert\ol\xi_{\ol\xi_{x_su_1}(\tau)(u_2,\ldots,u_N,u_\params)}(N\tau)-\ol\xi_{x_s(u_2,\ldots,u_N,u_\params)}(N\tau)\Vert^q\leq\beta(\Vert\ol\xi_{x_su_1}(\tau)-x_s\Vert^q,N\tau).
\end{align}
Hence, one gets
\begin{align}\label{upper_bound3}
\Vert\ol\xi_{\ol{H}_\params(x_\params)u_\params}(\tau)-\ol{H}_\params(x'_\params)\Vert\leq(\beta(\Vert\ol\xi_{x_su_1}(\tau)-x_s\Vert^q,N\tau))^{1/q}.
\end{align}
Since the inequality \eqref{upper_bound3} holds for all $x_\params\in X_\params$ and all $u_\params\in U_\params$, and $\beta$ is a $\mathcal{K}_\infty$ function with respect to its first argument when the second one is fixed, inequality \eqref{upper_bound1} holds.
\end{proof}

\begin{proof}{of Theorem \ref{main_theorem}:}
We start by proving that $R$ is an $\varepsilon$-approximate simulation relation from $S_{\tau}(\Sigma)$ to $\ol{S}_{\params}(\Sigma)$. Consider any \mbox{$\left(x_{\tau},x_{\params}\right)\in R$}. Condition (i) in Definition \ref{APSR} is satisfied because
\begin{equation}
\label{convexity}
(\mathbb{E}[\Vert x_{\tau}-\ol{H}_\params(x_{\params})\Vert^q])^{\frac{1}{q}}\leq(\underline\alpha^{-1}(\mathbb{E}[V(x_{\tau},\ol{H}_\params(x_{\params}))]))^{\frac{1}{q}}\leq\varepsilon.
\end{equation}
We used the convexity assumption of $\underline\alpha$ and the Jensen inequality \cite{oksendal} to show the inequalities in (\ref{convexity}). Let us now show that condition (ii) in Definition
\ref{APSR} holds. Consider any \mbox{$\upsilon_{\tau}\in {U}_{\tau}$}. Choose an input \mbox{$u_{\params}\in U_{\params}$} satisfying
\begin{equation}
\Vert \upsilon_{\tau}-u_{\params}\Vert_{\infty}=\Vert \upsilon_{\tau}(0)-u_{\params}(0)\Vert\leq\mu.\label{b01}
\end{equation}
Note that the existence of such $u_\params$ is guaranteed by $\mathsf{U}$ being a finite union of boxes and by
the inequality $\mu\leq\boxspan(\mathsf{U})$ which guarantees that $\mathsf{U}\subseteq\bigcup_{p\in[\mathsf{U}]_{\mu}}\mathcal{B}_{{\mu}}(p)$. Consider the transition \mbox{$x_{\tau}\rTo^{\upsilon_{\tau}}_{\tau} x'_{\tau}=\xi_{x_{\tau}\upsilon_{\tau}}(\tau)$} $\PP$-a.s. in $S_{\tau}(\Sigma)$. Since $V$ is a \mbox{$\delta$-ISS-M$_q$} Lyapunov function for $\Sigma$ and using inequality \eqref{b01}, we have (cf. equation (3.3) in \cite{majid8})
\begin{align}\label{b02}
\mathbb{E}[V(x'_{\tau},\xi_{\ol{H}_\params(x_{\params})u_{\params}}(\tau))] \leq \EE[V(x_\tau,\ol{H}_\params(x_q))] \mathsf{e}^{-\kappa\tau}+ \frac{1}{\mathsf{e}\kappa} \rho(\|\upsilon_{\tau}-u_{\params}\|_\infty)\leq \underline\alpha\left(\varepsilon^q\right) \mathsf{e}^{-\kappa\tau} + \frac{1}{\mathsf{e}\kappa}\rho(\mu).
\end{align}
Note that existence of $u_\params$, by the definition of $\ol{S}_\params(\Sigma)$, implies the existence of $x_{\params}\rTo^{u_{\params}}_{\params}x'_{\params}$ in $\ol{S}_{\params}(\Sigma)$.
Using Lemma \ref{lemma3}, the concavity of $\widehat\gamma$, the Jensen inequality \cite{oksendal}, equation \eqref{eta}, the inequalities (\ref{supplement}), (\ref{bisim_cond}), (\ref{b02}), and triangle inequality, we obtain
\begin{align*}
\mathbb{E}[V(x'_{\tau},\ol{H}_\params(x'_{\params}))]&=\mathbb{E}[V(x'_\tau,\xi_{\ol{H}_\params(x_{\params})u_{\params}}(\tau))+V(x'_{\tau},\ol{H}_\params(x'_\params))-V(x'_\tau,\xi_{\ol{H}_\params(x_{\params})u_{\params}}(\tau))]\\ \notag
&=  \mathbb{E}[V(x'_{\tau},\xi_{\ol{H}_\params(x_{\params})u_{\params}}(\tau))]+\mathbb{E}[V(x'_{\tau},\ol{H}_\params(x'_\params))-V(x'_\tau,\xi_{\ol{H}_\params(x_{\params})u_{\params}}(\tau))]\\\notag&\leq\underline\alpha(\varepsilon^q)\mathsf{e}^{-\kappa\tau}+\frac{1}{\mathsf{e}\kappa}\rho(\mu)+\mathbb{E}[\widehat\gamma(\Vert\xi_{\ol{H}_\params(x_{\params})u_{\params}}(\tau)-\ol{H}_\params(x'_{\params})\Vert)]\\\notag
&\leq\underline\alpha(\varepsilon^q)\mathsf{e}^{-\kappa\tau}+\frac{1}{\mathsf{e}\kappa}\rho(\mu)+\widehat\gamma(\mathbb{E}[\Vert\xi_{\ol{H}_\params(x_{\params})u_{\params}}(\tau)-\ol{\xi}_{\ol{H}_\params(x_{\params})u_{\params}}(\tau)+\ol{\xi}_{\ol{H}_\params(x_{\params})u_{\params}}(\tau)-\ol{H}_\params(x'_{\params})\Vert])\\\notag
&\leq\underline\alpha(\varepsilon^q)\mathsf{e}^{-\kappa\tau}+\frac{1}{\mathsf{e}\kappa}\rho(\mu)+\widehat\gamma(\mathbb{E}[\Vert\xi_{\ol{H}_\params(x_{\params})u_{\params}}(\tau)-\ol{\xi}_{\ol{H}_\params(x_{\params})u_{\params}}(\tau)\Vert]+\Vert\ol{\xi}_{\ol{H}_\params(x_{\params})u_{\params}}(\tau)-\ol{H}_\params(x'_{\params})\Vert)\\\notag
&\leq\underline\alpha(\varepsilon^q)\mathsf{e}^{-\kappa\tau}+\frac{1}{\mathsf{e}\kappa}\rho(\mu)+\widehat\gamma((h_{x_s}((N+1)\tau))^{\frac{1}{q}}+\eta)\leq\underline\alpha(\varepsilon^q).
\end{align*}
Therefore, we conclude that \mbox{$\left(x'_{\tau},x'_{\params}\right)\in{R}$} and that condition (ii) in Definition \ref{APSR} holds.

Now we prove that $R^{-1}$ is an
$\varepsilon$-approximate simulation relation from $\ol{S}_{\params}(\Sigma)$ to $S_{\tau}(\Sigma)$.
Consider any \mbox{$\left(x_{\tau},x_{\params}\right)\in R$} (or equivalently \mbox{$\left(x_{\params},x_{\tau}\right)\in R^{-1}$}). As showed in the first part of the proof, condition (i) in Definition \ref{APSR} is satisfied.
Let us now show that condition (ii) in Definition \ref{APSR} holds.
Consider any \mbox{$u_{\params}\in U_{\params}$}. Choose the input \mbox{$\upsilon_{\tau}=u_\params$} and consider \mbox{$x'_{\tau}=\xi_{x_{\tau}\upsilon_{\tau}}(\tau)$ $\PP$-a.s. in $S_{\tau}(\Sigma)$}.
Since $V$ is a \mbox{$\delta$-ISS-M$_q$} Lyapunov function for $\Sigma$, one obtains (cf. equation 3.3 in \cite{majid8}):
\begin{equation}
\mathbb{E}[V(x'_{\tau},\xi_{\ol{H}_\params(x_{\params})u_{\params}}(\tau))]\leq \mathsf{e}^{-\kappa\tau}\EE[V(x_{\tau},\ol{H}_\params(x_\params))]\leq\mathsf{e}^{-\kappa\tau}\underline\alpha\left(\varepsilon^q\right).\label{b03}%
\end{equation}
Using Lemma \ref{lemma3}, the definition of $\ol{S}_\params(\Sigma)$, the concavity of $\widehat\gamma$, the Jensen inequality \cite{oksendal}, equation \eqref{eta}, the inequalities (\ref{supplement}), (\ref{bisim_cond}), (\ref{b03}), and triangle inequality, we obtain
\begin{align*}\nonumber
\mathbb{E}[V(x'_{\tau},\ol{H}_\params(x'_{\params}))]&=\mathbb{E}[V(x'_{\tau},\xi_{\ol{H}_\params(x_{\params})u_{\params}}(\tau))+V(x'_{\tau},\ol{H}_\params(x'_{\params}))-V(x'_\tau,\xi_{\ol{H}_\params(x_{\params})u_{\params}}(\tau))]\\\notag &= \mathbb{E}[V(x'_{\tau},\xi_{\ol{H}_\params(x_{\params})u_{\params}}(\tau))]+\mathbb{E}[V(x'_{\tau},\ol{H}_\params(x'_{\params}))-V(x'_\tau,\xi_{\ol{H}_\params(x_{\params})u_{\params}}(\tau))]\\\notag&\leq\mathsf{e}^{-\kappa\tau}\underline\alpha(\varepsilon^q)+\mathbb{E}[\widehat\gamma(\Vert\xi_{\ol{H}_\params(x_{\params})u_{\params}}(\tau)-\ol{H}_\params(x'_{\params})\Vert)]\\\notag&\leq\mathsf{e}^{-\kappa\tau}\underline\alpha(\varepsilon^q)+\widehat\gamma(\mathbb{E}[\Vert\xi_{\ol{H}_\params(x_{\params})u_{\params}}(\tau)-\ol{\xi}_{\ol{H}_\params(x_{\params})u_{\params}}(\tau)+\ol{\xi}_{\ol{H}_\params(x_{\params})u_{\params}}(\tau)-\ol{H}_\params(x'_{\params})\Vert])\\\notag&\leq\mathsf{e}^{-\kappa\tau}\underline\alpha(\varepsilon^q)+\widehat\gamma(\mathbb{E}[\Vert\xi_{\ol{H}_\params(x_{\params})u_{\params}}(\tau)-\ol{\xi}_{\ol{H}_\params(x_{\params})u_{\params}}(\tau)\Vert]+\Vert\ol{\xi}_{\ol{H}_\params(x_{\params})u_{\params}}(\tau)-\ol{H}_\params(x'_{\params})\Vert)\\\notag&\leq\mathsf{e}^{-\kappa\tau}\underline\alpha(\varepsilon^q)+\widehat\gamma((h_{x_s}((N+1)\tau))^{\frac{1}{q}}+\eta)\leq\underline\alpha(\varepsilon^q).
\end{align*}
Therefore, we conclude that \mbox{$(x'_{\tau},x'_{\params})\in{R}$} (or equivalently \mbox{$\left(x'_{\params},x'_{\tau}\right)\in R^{-1}$}) and condition (ii) in Definition \ref{APSR} holds.
\end{proof}

\begin{proof}{of Theorem \ref{main_theorem2}:}
We start by proving that $R$ is an $\varepsilon$-approximate simulation relation from $S_{\tau}(\Sigma)$ to $\ol{S}_{\params}(\Sigma)$. Consider any \mbox{$\left(x_{\tau},x_{\params}\right)\in R$}. Condition (i) in Definition \ref{APSR} is satisfied by the definition of $R$. Let us now show that condition (ii) in Definition
\ref{APSR} holds. Consider any \mbox{$\upsilon_{\tau}\in {U}_{\tau}$}. Choose an input \mbox{$u_{\params}\in U_{\params}$} satisfying
\begin{equation}
\Vert \upsilon_{\tau}-u_{\params}\Vert_{\infty}=\Vert \upsilon_{\tau}(0)-u_{\params}(0)\Vert\leq\mu.\label{b10}
\end{equation}
Note that the existence of such $u_\params$ is guaranteed by $\mathsf{U}$ being a finite union of boxes and by
the inequality $\mu\leq\boxspan(\mathsf{U})$ which guarantees that $\mathsf{U}\subseteq\bigcup_{p\in[\mathsf{U}]_{\mu}}\mathcal{B}_{{\mu}}(p)$. Consider the transition \mbox{$x_{\tau}\rTo^{\upsilon_{\tau}}_{\tau} x'_{\tau}=\xi_{x_{\tau}\upsilon_{\tau}}(\tau)$} $\PP$-a.s. in $S_{\tau}(\Sigma)$. It follows from the $\delta$-ISS-M$_q$ assumption on $\Sigma$ and (\ref{b10}) that:
\begin{align}\label{b20}
\mathbb{E}[\Vert x'_{\tau}-\xi_{\ol{H}_\params(x_{\params})u_{\params}}(\tau)\Vert^q] \leq \beta(\EE[\Vert x_\tau-\ol{H}_\params(x_q)\Vert^q],\tau)+  \gamma(\|\upsilon_{\tau}-u_{\params}\|_\infty)\leq \beta(\varepsilon^q,\tau) + \gamma(\mu).
\end{align}
Existence of $u_\params$, by the definition of $\ol{S}_\params(\Sigma)$, implies the existence of $x_{\params}\rTo^{u_{\params}}_{\params}x'_{\params}$ in $\ol{S}_{\params}(\Sigma)$.
Using equation \eqref{eta}, the inequalities (\ref{mismatch1}), (\ref{bisim_cond2}), (\ref{b20}), and triangle inequality, we obtain
\begin{align*}
(\mathbb{E}[\Vert x'_{\tau}-\ol{H}_\params(x'_{\params})\Vert^q])^{\frac{1}{q}}&=(\mathbb{E}[\Vert x'_\tau-\xi_{\ol{H}_\params(x_{\params})u_{\params}}(\tau)+\xi_{\ol{H}_\params(x_{\params})u_{\params}}(\tau)-\ol{\xi}_{\ol{H}_\params(x_{\params})u_{\params}}(\tau)+\ol{\xi}_{\ol{H}_\params(x_{\params})u_{\params}}(\tau)-\ol{H}_\params(x'_\params)\Vert^q])^{\frac{1}{q}}\\ \notag
&\leq  (\mathbb{E}[\Vert x'_{\tau}-\xi_{\ol{H}_\params(x_{\params})u_{\params}}(\tau)\Vert^q])^{\frac{1}{q}}+(\mathbb{E}[\Vert\xi_{\ol{H}_\params(x_{\params})u_{\params}}(\tau)-\ol{\xi}_{\ol{H}_\params(x_{\params})u_{\params}}(\tau)\Vert^q])^{\frac{1}{q}}\\\notag&\hspace{1.1cm}+(\mathbb{E}[\Vert\ol{\xi}_{\ol{H}_\params(x_{\params})u_{\params}}(\tau)-\ol{H}_\params(x'_\params)\Vert^q])^{\frac{1}{q}}\\\notag&\leq(\beta(\varepsilon^q,\tau) + \gamma(\mu))^{\frac{1}{q}}+(h_{x_s}((N+1)\tau))^{\frac{1}{q}}+\eta\leq\varepsilon.
\end{align*}
Therefore, we conclude that \mbox{$\left(x'_{\tau},x'_{\params}\right)\in{R}$} and that condition (ii) in Definition \ref{APSR} holds.

Now we prove that $R^{-1}$ is an
$\varepsilon$-approximate simulation relation from $\ol{S}_{\params}(\Sigma)$ to $S_{\tau}(\Sigma)$.
Consider any \mbox{$\left(x_{\tau},x_{\params}\right)\in R$} (or equivalently \mbox{$\left(x_{\params},x_{\tau}\right)\in R^{-1}$}). Condition (i) in Definition \ref{APSR} is satisfied by the definition of $R$. Let us now show that condition (ii) in Definition \ref{APSR} holds.
Consider any \mbox{$u_{\params}\in U_{\params}$}. Choose the input \mbox{$\upsilon_{\tau}=u_\params$} and consider \mbox{$x'_{\tau}=\xi_{x_{\tau}\upsilon_{\tau}}(\tau)$ $\PP$-a.s. in $S_{\tau}(\Sigma)$}.
Since $\Sigma$ is \mbox{$\delta$-ISS-M$_q$}, one obtains:
\begin{equation}
\mathbb{E}[\Vert x'_{\tau}-\xi_{\ol{H}_\params(x_{\params})u_{\params}}(\tau)\Vert^q]\leq \beta(\EE[\Vert x_{\tau}-\ol{H}_\params(x_\params)\Vert^q],\tau)\leq\beta(\varepsilon^q,\tau).\label{b30}%
\end{equation}
Using definition of $\ol{S}_\params(\Sigma)$, equation \eqref{eta}, the inequalities \eqref{mismatch1}, (\ref{bisim_cond2}), (\ref{b30}), and the triangle inequality, we obtain
\begin{align*}
(\mathbb{E}[\Vert x'_{\tau}-\ol{H}_\params(x'_{\params})\Vert^q])^{\frac{1}{q}}&=(\mathbb{E}[\Vert x'_\tau-\xi_{\ol{H}_\params(x_{\params})u_{\params}}(\tau)+\xi_{\ol{H}_\params(x_{\params})u_{\params}}(\tau)-\ol{\xi}_{\ol{H}_\params(x_{\params})u_{\params}}(\tau)+\ol{\xi}_{\ol{H}_\params(x_{\params})u_{\params}}(\tau)-\ol{H}_\params(x'_\params)\Vert^q])^{\frac{1}{q}}\\ \notag
&\leq  (\mathbb{E}[\Vert x'_{\tau}-\xi_{\ol{H}_\params(x_{\params})u_{\params}}(\tau)\Vert^q])^{\frac{1}{q}}+(\mathbb{E}[\Vert\xi_{\ol{H}_\params(x_{\params})u_{\params}}(\tau)-\ol{\xi}_{\ol{H}_\params(x_{\params})u_{\params}}(\tau)\Vert^q])^{\frac{1}{q}}\\\notag&\hspace{1.1cm}+(\mathbb{E}[\Vert\ol{\xi}_{\ol{H}_\params(x_{\params})u_{\params}}(\tau)-\ol{H}_\params(x'_\params)\Vert^q])^{\frac{1}{q}}\\\notag&\leq(\beta(\varepsilon^q,\tau))^{\frac{1}{q}}+(h_{x_s}((N+1)\tau))^{\frac{1}{q}}+\eta\leq\varepsilon.
\end{align*}
Therefore, we conclude that \mbox{$(x'_{\tau},x'_{\params})\in{R}$} (or equivalently \mbox{$\left(x'_{\params},x'_{\tau}\right)\in R^{-1}$}) and condition (ii) in Definition \ref{APSR} holds.
\end{proof}

\begin{proof}{of Theorem \ref{thm:number.samples}:}
  Denote $\hat\theta := \theta - r/2>0$,
  and $\mathbf{d}_M(a):=\left(\frac1M \sum\limits_{i=1}^M \|\xi^i_{x_s x_\q}-a\|^q\right)^{\frac1q}$ for all $a\in \R^n$.
  It follows from \cite[Theorem 4.5.4]{kloeden} that for all $p\geq 1$ and $a\in \R^n$
  \begin{equation*}
    \EE\left[\|\xi_{x_sx_\q}(N\tau) - a\|^p\right] \leq b(a, p).
  \end{equation*}
  Hence, by Chernoff's inequality for any $a'\in A^r$ we obtain:
  \begin{equation*}
    \PP\left(\left|\left(\mathbf{d}(H_\q(x_\q),a')\right)^q - (\mathbf{d}_M(a'))^q\right|\geq \hat\theta\right)
    \leq \frac{b(a', 2q)}{M\hat\theta^2}.
  \end{equation*}
  Furthermore, since $x\mapsto x^q$ is H\"older continuous with power $q$,  
  \begin{equation*}
    \PP\left(\left|\mathbf{d}(H_\q(x_\q),a') - \mathbf{d}_M(a')\right|\geq \hat\theta\right)
    \leq \frac{b(a', 2q)}{M\hat\theta^{2q}}.
  \end{equation*}
  Thus, for the union of such events over $a'\in A^r$, we have
  \begin{align}\label{eq:proof.1}
    &\PP\left(\exists a'\in A^r \text{ s.t. }\left|\mathbf{d}(H_\q(x_\q),a') - \mathbf{d}_M(a')\right|\geq \hat\theta\right)\leq \frac{|A^r|b(a^*, 2q)}{M\hat\theta^{2q}},
  \end{align}
  due to the fact that the probability of a union is dominated by the sum of probabilities.
  Let $[\cdot]:A\to A^r$ be any surjective map such that $\left\Vert a-[a]\right\Vert\leq r/2$ for all $a\in A$, i.e. $[\cdot]$ chooses an $r/2$-close point in the grid $A^r$.
  Using this map,
  we can extrapolate the inequality \eqref{eq:proof.1} to the whole set $A$ since
  \begin{align*}
    \left\vert\mathbf{d}(H_\q(x_\q),a) - \mathbf{d}_M([a])\right\vert &\leq \left\vert\mathbf{d}(H_\q(x_\q),a) - \mathbf{d}(H_\q(x_\q),[a])\right\vert + \left\vert\mathbf{d}(H_\q(x_\q),[a]) - \mathbf{d}_M([a])\right\vert
    \\
    &\leq r/2 +\left\vert\mathbf{d}(H_\q(x_\q),[a]) - \mathbf{d}_M([a])\right\vert,
  \end{align*}
  where we used the fact that $|\mathbf{d}(H_\q(x_\q),a) - \mathbf{d}(H_\q(x_\q),[a])|\leq\|a - [a]\|$ by the triangle inequality.
  As a result, the following inequality holds:
  \begin{equation}\label{eq:proof.2}
  \begin{split}
    \PP\left(\exists a\in A \text{ s.t. }\left\vert\mathbf{d}(H_\q(x_\q),a) - \mathbf{d}_M([a])\right\vert\geq \theta\right)
    \leq \PP\left(\exists a'\in A^r\text{ s.t. }\left\vert\mathbf{d}\left(H_\q(x_\q),a'\right) - \mathbf{d}_M\left(a'\right)\right\vert\geq \hat\theta\right).
  \end{split}
  \end{equation}
  On the other hand,
  since for any two functions $f,g:A\to \R$ it holds that
  \begin{equation*}
    \left\vert\inf_{a\in A}f(a) - \inf_{a\in A}g(a)\right\vert\leq\sup_{a\in A}|f(a) - g(a)|,
  \end{equation*}
  we obtain that
  \begin{align*}
    &\PP\left(\left\vert\mathbf{d}(H_\q(x_\q),A) - \mathbf{d}^r_M\right\vert\geq \theta\right)\leq \PP\left(\exists a\in A \text{ s.t. }\left\vert\mathbf{d}(H_\q(x_\q),a) - \mathbf{d}_M([a])\right\vert\geq \theta\right).
  \end{align*}
  Combining the latter inequality with \eqref{eq:proof.1} and \eqref{eq:proof.2} yields:
  \begin{equation*}
    \PP\left(\left\vert\mathbf{d}(H_\q(x_\q),A) - \mathbf{d}^r_M\right\vert\geq \theta\right) \leq \frac{|A^r|b(a^*, 2q)}{M\hat\theta^{2q}},
  \end{equation*}
  and in case $M$ satisfies the assumption of the theorem,
  the right-hand side is bounded above by $\pi$ as desired.
\end{proof}

\end{document}